\documentclass[11pt]{article}
\pdfoutput=1
\usepackage[utf8]{inputenc}
\usepackage[margin=1in]{geometry}
\usepackage{multirow}
\usepackage{amssymb,amsmath,amsthm,amsfonts,amsmath}
\usepackage{thmtools, thm-restate}
\usepackage{bm}
\usepackage{textgreek}
\usepackage{mathtools}
\usepackage{enumitem}
\usepackage[numbers,comma,sort&compress]{natbib}
\usepackage{graphicx}
\usepackage[font=small]{caption}
\usepackage[labelformat=simple]{subcaption}
\usepackage{float}
\usepackage[ruled,vlined,linesnumbered]{algorithm2e}
\usepackage{physics}
\usepackage{footnote}
\usepackage{xcolor}
\usepackage{mathrsfs}
\usepackage{mdframed}
\usepackage{bbm}
\usepackage{csquotes}
\usepackage{makecell}
\usepackage{booktabs,multirow,makecell}
\usepackage{quantikz}
\usepackage{tikz}
\usepackage{comment}
\usepackage{braket}
\usepackage[pagebackref]{hyperref}
\hypersetup{colorlinks=true,urlcolor=blue,linkcolor=blue,citecolor=[rgb]{.0,.50,.32},}
\usepackage{mathtools}
\usepackage{microtype}
\DeclarePairedDelimiter\ceil{\lceil}{\rceil}

\renewcommand{\backref}[1]{}

\renewcommand{\backrefalt}[4]{%
\ifcase #1 %
\or 
[p.\ #2]%
\else 
[pp.\ #2]%
\fi}

\newtheorem{theorem}{Theorem}
\newtheorem{definition}[theorem]{Definition}
\newtheorem{lemma}[theorem]{Lemma}
\newtheorem{proposition}[theorem]{Proposition}

\newtheorem{fact}[theorem]{Fact}

\usepackage[capitalise,nameinlink]{cleveref}
\Crefname{lemma}{Lemma}{Lemmas}
\Crefname{fact}{Fact}{Facts}
\Crefname{theorem}{Theorem}{Theorems}
\Crefname{corollary}{Corollary}{Corollaries}
\Crefname{claim}{Claim}{Claims}
\Crefname{example}{Example}{Examples}
\Crefname{problem}{Problem}{Problems}
\Crefname{definition}{Definition}{Definitions}
\Crefname{notation}{Notation}{Notations}
\Crefname{assumption}{Assumption}{Assumptions}
\Crefname{subsection}{Section}{Sections}
\Crefname{section}{Section}{Sections}
\Crefname{table}{Table}{Tables}
\Crefname{equation}{Eq.}{Eqs.}

\def\>{\rangle}
\def\<{\langle}

\newcommand{\eps}{\epsilon}
\newcommand{\OR}{\mathsf{OR}}
\newcommand{\AND}{\mathsf{AND}}
\newcommand{\XOR}{\mathsf{XOR}}
\newcommand{\INC}{\mathsf{INC}}
\newcommand{\Toff}{\mathrm{Toff}}

\newcommand{\B}{\{0,1\}}
\newcommand{\id}{I}

\newcommand{\E}{\mathbb{E}}
\newcommand{\N}{\mathbb{N}}

\newcommand{\Tadapt}{\mathcal{T}^{\mathrm{adaptive}}}
\newcommand{\Tmixed}{\mathcal{T}^{\mathrm{mixed}}}
\newcommand{\Tmix}{\mathcal{T}^{\mathrm{mixed}}}
\newcommand{\Tunitary}{\mathcal{T}^{\mathrm{unitary}}}
\newcommand{\Tunit}{\mathcal{T}^{\mathrm{unitary}}}
\newcommand{\Dd}{D_\diamond}
\newcommand{\PDTna}{\mathrm{PDT^{na}}}
\newcommand{\RPDTna}{\mathrm{RPDT}^{\mathrm{na}}_\epsilon}
\newcommand{\NAPDT}{\mathrm{PDT}^{\mathrm{na}}}
\newcommand{\RNAPDT}{\mathrm{RPDT}^{\mathrm{na}}}
\newcommand{\gatePDT}{\mathrm{gate}\PDTna}
\newcommand{\gateRPDT}{\mathrm{gate}\RNAPDT}

\DeclareMathOperator{\poly}{poly}
\DeclareMathOperator{\sign}{sign}
\DeclareMathOperator{\post}{post}

\title{Multi-qubit Toffoli with exponentially fewer $T$ gates}

\author{David Gosset\thanks{Google Quantum AI.}~${}^{,}$\thanks{Department of Combinatorics and Optimization and Institute for Quantum Computing, University of Waterloo.}~${}^{,}$\thanks{Perimeter Institute for Theoretical Physics.}\and Robin Kothari${}^*$\and Chenyi Zhang${}^{*,}$\thanks{Stanford University.}
}
\date{}

\begin{document}

\maketitle

\begin{abstract}

Prior work of Beverland et al.~\cite{beverland2020lower} has shown that any exact Clifford+$T$ implementation of the $n$-qubit Toffoli gate must use at least $n$ $T$ gates. Here we show how to get away with exponentially fewer $T$ gates, at the cost of incurring a tiny $1/\mathrm{poly}(n)$ error that can be neglected in most practical situations. More precisely, the $n$-qubit Toffoli gate can be implemented to within error $\epsilon$ in the diamond distance by a randomly chosen Clifford+$T$ circuit with at most $O(\log(1/\epsilon))$ $T$ gates. We also give a matching $\Omega(\log(1/\epsilon))$ lower bound that establishes optimality, and we show that any purely unitary implementation achieving even constant error must use $\Omega(n)$ $T$ gates.  We also extend our sampling technique to implement other Boolean functions. Finally, we describe upper and lower bounds on the $T$-count of Boolean functions in terms of non-adaptive parity decision tree complexity and its randomized analogue.
\end{abstract}

\section{Introduction}

Clifford circuits---that is, quantum computations that can be expressed as a sequence of single-qubit Hadamard, phase, and CNOT gates applied to a computational basis state---are efficiently classically simulable via the Gottesman-Knill theorem. They define an extraordinary classical limit of many-body quantum mechanics. In order to perform universal quantum computation, one requires non-Clifford resources, or \textit{magic}. This can be in the form of  a non-Clifford unitary or initial state. A natural choice is to augment the Cliffords with the single-qubit $T=\mathrm{diag}(1,e^{-i\pi/4})$ gate. The resulting Clifford+$T$ gate set is the canonical instruction set for fault-tolerant quantum computation in architectures based on the surface code, where Clifford gates can be performed fault-tolerantly directly while $T$ gates are performed via magic state injection \cite{bravyi2005universal} or other more complex methods \cite{gidney2024magic}.

Developing and optimizing techniques for compiling circuits over the Clifford+$T$ gate set is a fundamental task that has the potential to reduce the resource costs of implementing quantum algorithms in fault tolerant architectures. For example, asymptotically optimal and ancilla-free single-qubit compilation techniques for Clifford+$T$ due to Ross and Selinger \cite{ross2014optimal} represent a significant practical improvement over the general methods provided by the Solovay-Kitaev theorem.

Here we consider the task of implementing a given target unitary using as few $T$ gates as possible. The number of $T$ gates required---its \textit{$T$-count}---is a measure of the magic possessed by the unitary. It determines the hardness of classically simulating the unitary via the so-called stabilizer rank based methods \cite{bravyi2016trading, bravyi2016improved}. For few-qubit unitaries, where the size of the Clifford group is a reasonably small constant, it is also a proxy for the total gate count of an implementation.  The $T$-count dominates the total cost of fault-tolerant implementations based on magic state distillation.\footnote{A recent ``magic state cultivation" technique results in a different accounting that challenges this narrative for some fault-tolerant architectures~\cite{gidney2024magic}.}

We will demonstrate that the number of $T$ gates required to implement certain elementary multi-qubit operations can be far lower than previously thought if we allow a small amount of error.
 
One definition of the $T$-count of a unitary $U$, which we call unitary $T$-count (\Cref{def:unitaryT}), is the minimum number of $T$ gates in a circuit $C$ such that the unitary implemented by $C$ is $\eps$-close to $U$. (We allow error since most unitaries cannot be exactly implemented by a Clifford+$T$ circuit.)

However, it has long been known that taking probabilistic mixtures of unitary Clifford+$T$ circuits can often yield more efficient circuits~\cite{Cam17,Has17}. The mixed unitary $T$-count (\Cref{def:mixedT}) of $U$ is the minimum $k$ such that there is a channel $\Phi$ that is $\eps$-close to $U$ (in diamond distance) and is a probabilistic mixture over unitary Clifford+$T$ circuits of $T$-count at most $k$. 

Implementing a \emph{mixed} Clifford+$T$ circuit on a quantum computer requires no additional quantum hardware from the quantum computer: All the additional work is done by the (classical) compiler. 
The compiler samples a unitary Clifford+$T$ circuit from the probability distribution and outputs the circuit to be run on the quantum computer. If the original circuit contains multiple copies of $U$, the classical computer uses fresh samples for each copy. 

An even stronger model, which we do not use in any of our algorithms, is the model we call adaptive Clifford+$T$ circuits (\Cref{def:adaptiveT}). Here the algorithm may use mixtures, perform mid-circuit measurements, and use classical feed-forward (i.e., future gates in the quantum circuit may depend on past measurement outcomes). Note that this model assumes the quantum hardware is capable of mid-circuit measurements and classical feed-forward, which is not supported by all current hardware, although we expect that a fault-tolerant quantum computer will have this ability since it is required to perform quantum error-correction.

We denote the $T$-count in each of these models by $\Tunit_\eps(U) \geq \Tmix_\eps(U) \geq \Tadapt_\eps(U)$ respectively. As an example of the difference in power, consider the $T$-count of a typical single-qubit diagonal unitary. It has been demonstrated heuristically that (see  \cite[Table 1]{KLMPP23}), unitary, mixed, and adaptive Clifford+$T$ circuits can approximate such a unitary with $T$-count $3\log(1/\eps)$, $1.5\log(1/\eps)$, and $0.5\log(1/\eps)$ respectively.\footnote{Here and throughout this paper, all logarithms are computed base $2$} 

\paragraph{Multi-qubit Toffoli.}
We use the power of mixed Clifford+$T$ circuits to obtain dramatic improvements in $T$-count, well beyond constant factors. 
We first consider the $n$-qubit Toffoli gate:
\begin{equation}
\Toff_n|x\rangle|b\rangle = |x\rangle|b\oplus (x_1 \wedge \cdots \wedge x_{n-1})\rangle,\quad\textrm{for all}~x\in \{0,1\}^{n-1}~\textrm{and}~b\in \{0,1\},
\end{equation}
which reversibly computes the $\AND$ of the first $n-1$ bits into the last register. This gate, and gates that are Clifford-equivalent to it, is a central building block in quantum algorithms. Note that if we conjugate the last qubit by Hadamard, we obtain the $n$-qubit controlled $Z$ gate, which acts as
\begin{equation}
C^{n-1}Z|x\rangle=(-1)^{x_1x_2\dots x_n}|x\rangle.
\end{equation}
By conjugating by a full layer of Hadamards, we can also get the diffusion operator in Grover's algorithm, which reflects about the uniform superposition state.

It is well known that $\Toff_n$ (or equivalently $C^{n-1}Z$) can be implemented exactly using only $O(n)$ $\Toff_3$ gates~\cite{BBC+95}, and each $\Toff_3$ gate can be implemented exactly by a unitary Clifford+$T$ circuit using 7 $T$ gates~\cite{NC10}. This shows that
\begin{equation}
    \Tunit_0(\Toff_n) = O(n).
\end{equation}

On the other hand, Beverland, Campbell, Howard, and Kliuchnikov~\cite[Proposition 4.1]{beverland2020lower} give a matching lower bound, even in the stronger adaptive model:
\begin{equation}
\Tadapt_0(\Toff_n)\geq n.
\label{eq:zeroerror}
\end{equation}

This seemingly closes the question (up to constant factors), since the upper and lower bounds match asymptotically.

The starting point of our work is the observation that the above lower bound only applies to the zero-error setting, whereas for practical applications, some error is always acceptable. In many cases, inverse polynomial error $\epsilon=1/\mathrm{poly}(n)$ is more than enough. Surprisingly, we find that approximately implementing $\Toff_n$ to within such an error budget is vastly cheaper than implementing it exactly. In particular, $\Tmix_{1/\poly(n)}(\Toff_n) \leq \Tunit_0(\Toff_{O(\log n)}) = O(\log n)$. More generally we show the following.

\begin{restatable}{theorem}{Toffupper}\label{thm:Toffupper}
For any positive integer $n$ and $\epsilon>0$, we have
\begin{equation}
\Tmix_\eps(\Toff_n) \leq \Tunit_0(\Toff_{\lceil \log(1/\epsilon)\rceil+3})=O(\log(1/\epsilon)).
\label{eq:tupperbnd}
\end{equation}  
\end{restatable}

Thus the cost of $\eps$-approximating an $n$-qubit Toffoli with a mixed Clifford+$T$ circuit is at most the cost of exactly implementing one small Toffoli on $\lceil \log(1/\epsilon)\rceil+3$ qubits, which is independent of $n$!
Thus we get to replace a large Toffoli with a small Toffoli, and this is advantageous whenever $n \geq \lceil \log(1/\epsilon)\rceil+3$, which fails to hold only if $\eps$ is exponentially small in $n$.

\Cref{thm:Toffupper} also yields the same upper bound for $C^n X$, $C^n Z$, and the Grover diffusion operator, which are Clifford-equivalent to the multi-qubit Toffoli gate. We also get the same upper bound for $C^n G$ for any single qubit gate $G$, by noting that $C^n G$ can be implemented using two $C^n X$ gates and one controlled-$G$ gate, which can be implemented with $O(\log(1/\eps))$ $T$ gates unitarily~\cite{ross2014optimal}. More generally, for any unitary $U$, we get $\Tmix_\eps(C^n U) \leq \Tmix_{\eps/2}(C~ U)+O(\log(1/\eps))$.

The method we use to approximate $\Toff_n$ is quite simple. As noted above, $\Toff_n$ reversibly computes the $n-1$-bit $\AND$ function, denoted $\AND_{n-1}$. It will be slightly more convenient to consider the Clifford-equivalent gate
$X^{\otimes n-1} \Toff_n X^{\otimes n}$  which reversibly computes the $\OR$ function on $n-1$ bits, denoted $\OR_{n-1}$. 

We now show how to approximate the $\OR_n$ gate using parity functions of the form $\XOR_S(x)=\bigoplus_{i \in S} x_i$, which are Clifford gates, and one small $\OR$ gate.
First, we observe two facts: 
\begin{itemize}
    \item If $\OR_n(x)=0$ then $\XOR_S(x)=0$ for all $S \subseteq [n]$. 
    \item If $\OR_n(x)=1$ then $\XOR_S(x)=0$ for exactly half the subsets $S \subseteq[n]$.
\end{itemize}
So if we pick a random $S \subseteq[n]$,\footnote{The idea to compute an $\OR$ using random $\XOR$s is a classic algorithmic technique in computer science. It is used to show that the public-coin randomized communication complexity of the equality function is $O(1)$~\cite[Example 3.13]{KN96} and is an example of the algorithmic technique known as randomized fingerprinting~\cite[Chapter 7]{MR95}.} 
the function $\XOR_S(n)$ is already a constant-error approximation to the $\OR_n$ function.\footnote{In fact, this is a \emph{one-sided error} approximation, which means that on one type of input, the inputs that evaluate to $0$, the approximation is always correct, and the error only occurs on the other type of input.} Now we only need to boost the success probability to $1-\eps$.

To make this an approximation with error $\eps$, we sample $k$ uniformly random subsets $S_1,\ldots, S_k \subseteq [n]$. Now if $\OR_n(x)=0$, then $\XOR_{S_i}(x)=0$ for all of these subsets. On the other hand, if $\OR_n(x)=1$, then the probability that \emph{all} $k$ of these subsets have $\XOR_{S_i}(x)=0$ is $1/2^{k}$. 
So if we define $g_{S_1,\ldots,S_k}(x)=\OR_k (\XOR_{S_1}(x),\XOR_{S_2}(x),\ldots, \XOR_{S_k}(x))$, then for any $x \in \{0,1\}^n$
\begin{equation}
\Pr_{S_1,\ldots,S_k}[\OR_n(x) \neq g_{S_1,\ldots,S_k}(x)] \leq 1/2^k.
\label{eq:probbound}
\end{equation}
Choosing $k=\ceil{\log(1/\eps)}$ ensures that the overall error is at most $\eps$. 

To move from approximating a Boolean function to approximating a unitary ($\Toff_n$), we need some notation. Throughout this paper, for any Boolean function $f:\{0,1\}^n \to \{0,1\}$, let $U_f$ be the unitary that reversibly computes $f$ as follows:
\begin{equation}\label{eq:Uf}
U_f|x\rangle|b\rangle=|x\rangle|b\oplus f(x)\rangle \quad\text{for all}~~x\in \{0,1\}^n ~~\text{and}~~b\in \{0,1\}.
\end{equation}
Now $\Toff_n = U_{\AND_{n-1}}$ and $(X^{\otimes n}\otimes \id) \Toff_{n+1} X^{\otimes n+1} = U_{\OR_n}$. $U_{\OR_n}$ is now $\eps$-approximated by the mixed unitary circuit that first picks $S_1,\ldots,S_k$ uniformly at random and then applies the unitary $U_g$ corresponding to $g_{S_1,\ldots,S_k}(x)$. Note that the only non-Clifford gate here is the $\OR_k$ gate, which is implemented using one $\Toff_{k+1}$ gate.

We highlight a few interesting features of this approximation: Observe that the Toffoli gate is approximated by a distribution over gates of the form $U_g$, but each of these gates individually is perfectly distinguishable from the Toffoli gate (i.e., they are distance $1$ from the Toffoli gate). So although none of the gates in the distribution is close to the Toffoli gate, the mixture is $\eps$-close to it. This phenomenon is unique to implementing unitaries and does not occur with state preparation. If a distribution over states is $\eps$-close (in trace distance) to a pure state, then at least one of the states in the support of the distribution is also $O(\sqrt{\eps})$-close (in trace distance) to the pure state.

Another consequence of this approximation is that the Clifford hierarchy is very non-robust to error. The $\Toff_n$ gate is in level $n+1$ of the Clifford hierarchy~\cite{CGK17}, but we show that it can be $\eps$-approximated by gates in level $O(\log(1/\eps))$ of the Clifford hierarchy. 

Lastly, our result may have applications to learning and classical simulation algorithms that work when the $T$-count is low, since we now show that even circuits with large Toffoli gates do effectively have low $T$-count.

\paragraph{Optimality.}
Some natural questions arise about the optimality of our upper bound. First, one might ask if the mixed Clifford+$T$ model is necessary at all to achieve this result. Could it be possible that even $\Tunit_\eps(\Toff_n)$ is small? The lower bound of \cite{beverland2020lower} only says that achieving $\eps=0$ requires large $T$ count. Our first lower bound establishes that unitary circuits approximating $\Toff_n$ must use $\Omega(n)$ gates:

\begin{restatable}{theorem}{Toffunitarylower}\label{thm:Toffunitarylower}
    For any $\epsilon\in[0,1/2)$ and large enough $n$, we have
\begin{equation}
\Tunitary_\eps(\Toff_n) \geq n-2.
\end{equation} 
\end{restatable}

Another natural question is whether one can do better than the upper bound in \Cref{thm:Toffupper}. We provide a matching lower bound, using a generalization of the stabilizer nullity technique of Ref.~\cite{beverland2020lower}, showing this is impossible even in the more powerful adaptive Clifford+$T$ model.

\begin{restatable}{theorem}{Toffadaptivelower}\label{thm:Toffadaptivelower}
For large enough $n$ and $1/\eps$, we have
\begin{equation}
\Tmix_\eps(\Toff_n) \geq \Tadapt_\eps(\Toff_n) = \Omega(\min\{n,\log(1/\epsilon)\}).
\end{equation}
\end{restatable}

\paragraph{Generalization.} 
\Cref{thm:Toffupper} shows how to implement $\Toff_n = U_{\AND_{n-1}}$ or its Clifford-equivalent $U_{\OR_{n-1}}$ efficiently with low $T$-count. We give an explanation for this in terms of Fourier expansion of the associated Boolean function.  Recall the Boolean Fourier expansion of $f:\B^n \to \B$ as a linear combination of parities:
\begin{equation}\label{eq:fourier}
f(x)=\sum_{S\subseteq [n]}\widehat{f}(S) \chi_S(x),
\end{equation}
where $\chi_S(x)=(-1)^{\XOR_S(x)}$. Let us define the Fourier 1-norm of a function $f$ as $\|\widehat f\|_1\equiv \sum_{S\subseteq [n]}|\widehat f(S)|$. 
The reason we were able to approximate $\Toff_n$ has to do with the $\OR$ function having small Fourier 1-norm since 
\begin{equation}
\OR_n(x)=1-\frac{1}{2^n}\sum_{S\subseteq [n]}\chi_S(x),
\end{equation}
and therefore $\|\widehat{\OR}_n\|_1\leq 2$. We generalize \Cref{thm:Toffupper} to all functions with small Fourier 1-norm.

\begin{restatable}{theorem}{FourierOneNorm}\label{thm:fourier1norm}
Let $f:\{0,1\}^n\rightarrow \{0,1\}$ and $\epsilon>0$ be given. Then
\begin{equation}\label{eqn:fourier-1-intro}
\Tmix_\eps(U_f) = O(\|\widehat f\|_1^2\log(1/\epsilon)).
\end{equation} 
\end{restatable}

\Cref{thm:fourier1norm} has the advantage that the Fourier 1-norm  is relatively easy to work with---in particular, we can analytically understand its scaling with $n$ for most functions of interest. However we do not expect that it tightly characterizes the $T$-count of mixed Clifford+$T$ circuits. 

The algorithm to establish \Cref{thm:fourier1norm} is also simple.
We can see that \Cref{eq:fourier} is proportional to a (signed) average of  parity functions with respect to a probability distribution $p(S)\equiv |\widehat{f}(S)|/\|\widehat{f}\|_1$. In order to approximate $U_f$, we sample $k$ sets $S$ from this distribution: $S_1,S_2,\ldots, S_k$. Next  we compute the sample mean 
\begin{equation}
\tilde{g}(x) = \frac{\|\widehat{f}\|_1}{k}\sum_{i=1}^{k} \mathrm{sign}(\widehat{f}(S))\chi_{S_i}(x),
\label{eq:samplemean}
\end{equation}
and we define a Boolean function $g(x)$ which is $1$ iff $\tilde{g}(x)\geq 1/2$. We then reversibly compute $g$. In \Cref{sec:boolean} we show that $k=O(\|\widehat{f}\|_1^2\log(1/\epsilon))$ suffices to approximate $U_f$ to within error $\epsilon$, and we show how to implement this procedure with $O(k)$ $T$ gates.

\paragraph{Parity decision trees.} As our final structural result, we show how to upper and lower bound the $T$-count of Boolean functions computed by unitary and mixed circuits using  non-adaptive parity decision tree complexity and its randomized analogue. 

A function $f:\B^n \to \B$ has non-adaptive parity decision tree complexity at most $k$, which we denote by $\PDTna(f)\leq k$, if there exist $k$ subsets $S_1,\ldots,S_k\subseteq [n]$, such that $f(x)=g(\XOR_{S_1}(x),\ldots,\XOR_{S_k}(x))$ for an arbitrary fixed function $g:\B^k\to\B$. The non-adaptive randomized parity decision tree complexity of $f$, $\RPDTna(f)$, is defined analogously by taking a probability distribution over non-adaptive parity decision trees, such that for any input $x\in\B^n$ the output is correct with probability at least $1-\eps$.\footnote{Non-adaptive randomized parity decision tree complexity is also called randomized linear sketch complexity in the sketching literature~\cite{KMSY18}.}

We also define the (non-standard) measure $\gatePDT(f)$ to refer to the classical gate complexity of the function $g$ in the definition of $\NAPDT(f)$; $\gateRPDT(f)$ is defined analogously. We show that these measures upper and lower bound $T$-count. 
\begin{restatable}{theorem}{PDTintro}\label{thm:PDT-thm-intro}
    For any Boolean function $f:\B^n \to \B$ and any $\eps\geq 0$,
    \begin{align}
        \PDTna(f)-1 &\leq \Tunit_{1/3}(U_f) = O(\gatePDT(f)),\label{eq:pdtf} ~~\text{and}\\
        \RPDTna(f)-1 &\leq \Tmix_{\eps}(U_f) = O(\gateRPDT_{\eps}(f)).
        \label{eq:rpdtf}
    \end{align}  
\end{restatable}

This shows that for some functions of interest, such as $\AND$ and $\OR$, whose
non-adaptive randomized parity decision trees have the same gate complexity as their decision tree complexity, we obtain a tight characterization of their $T$ count. Note that 
$\gateRPDT(f)$ can also be upper bounded by the right-hand side of \Cref{eqn:fourier-1-intro}  with essentially the same argument.

Since $\Tmix_\eps(U_f) \geq \Tadapt_\eps(U_f)$, a natural question is whether $\Tadapt_\eps(U_f)$ can similarly be lower bounded by the LHS of \Cref{eq:rpdtf}. As a first step in this direction, we establish this  lower bound in the special case $\epsilon=0$:

\begin{restatable}{theorem}{PDTadapt}
\label{thm:adptive-lower-PDT}
For any Boolean function $f\colon \B^n\to \B$, we have
\begin{align}
\Tadapt_0(U_f)\geq \NAPDT(f)-1.
\end{align}
\end{restatable}

Extending this lower bound to the case $\epsilon>0$ is left as a challenge for future work. 

We note that our proof of \Cref{thm:PDT-thm-intro} given in \Cref{sec:pdt} establishes a slightly stronger result than the one stated above:  we show that the lower bounds in \Cref{eq:pdtf} and \Cref{eq:rpdtf} hold even for unitary or mixed quantum circuits (respectively) that are provided with an ancilla register prepared in an arbitrary advice state (rather than the all-zeros computational basis state). In contrast, the lower bound in \Cref{thm:adptive-lower-PDT} cannot be strengthened in a similar fashion;  if we provide an \textit{adaptive} Clifford+$T$ circuit with a suitable advice state (several copies of the single-qubit magic state) then we can compute any Boolean function exactly with no $T$ gates via magic state injection.

\paragraph{Applications.}
We then apply these techniques to some Boolean functions of interest to establish upper and lower bounds. This is summarized in \Cref{tab:function_costs}. 
\begin{table}[h!]
\centering
\begin{tabular}{p{31em}l}
\toprule
\textbf{Function $f$} & $\Tmix_\eps(U_f)$ \\
\midrule
$\OR_n(x)$: Logical $\OR$ of an $n$-bit input string $x$ & $O(\log(1/\eps))$ \\
$\mathsf{HW}_n^d(x)$: Is the Hamming weight of $x\in\B^n$, $|x|\leq d$ for constant $d$? & $O(\log(1/\eps))$ \\
$\mathsf{HW}_n^{k,2k}(x)$: For $x\in\B^n$ and $k\in[n]$, is $|x|\leq k$ or $|x|\geq 2k$? & $O(\log(1/\eps))$ \\
$\mathsf{CW}_n^C(x)$: For a fixed linear code $C\subseteq\B^n$, is $x\in C$?  & $O(\log(1/\eps))$\\
$\mathsf{MEQ}_{n,m}(M)$: Does $M \in \B^{n\times m}$ have identical rows? & $O(\log(1/\eps))$ \\
$\mathsf{RankOne}_{n,m}(M)$: Does $M \in \B^{n\times m}$ have rank 1 (over $\mathbb{F}_2$)? & $O(\log(1/\eps))$ \\
\midrule[\heavyrulewidth]
$\mathsf{GT}_n(x,y)$: Is $x\in\{0,1\}^n$ greater than $y\in\{0,1\}^n$? & $\Omega(n)$\\
$\mathsf{INC}_n(x)$: Given $x\in\{0,1\}^n$, output $x+1\mod2^n$ & $\Omega(n)$\\
$\mathsf{ADD}_n(x,y)$: Given $x,y\in\{0,1\}^n$, output $x+y\mod2^n$ & $\Omega(n)$\\
$\mathsf{MAJ}_n(x)$: Is the Hamming weight of $x$ greater than $n/2$? & $\Omega(n)$\\
\bottomrule
\end{tabular}
\caption{The $T$-count to approximately implement some Boolean functions}
\label{tab:function_costs}
\end{table}

Note that it was previously shown by Beverland et al. \cite{beverland2020lower} that $\mathsf{ADD}_n$ and $\Toff_n$ both require $\Omega(n)$ $T$ gates to implement with \textit{zero error}, and we show that the two gates have dramatically different cost in the presence of error. 

\paragraph{Concurrent work.}
A concurrent work of Uma Girish, Alex May, Natalie Parham, and Henry Yuen has established similar lower bounds on the unitary and mixed $T$-count of Boolean functions in terms of notions from communication complexity, as well as a lower bound in an adaptive model that differs from ours. We are grateful to Alex May for a discussion in which we learned that their results hold in the presence of an advice state; after that discussion, we noted that our lower bounds from \Cref{thm:PDT-thm-intro} also hold in the presence of an advice state.

\paragraph{Paper organization.}
The remainder of this paper is organized as follows. 
In \Cref{sec:CliffordT} we define the distance measures used and the models of Clifford+$T$ circuits we study.
In \Cref{sec:mct} we discuss the multi-qubit Toffoli gate and prove \Cref{thm:Toffupper} and \Cref{thm:Toffadaptivelower}. 
In \Cref{sec:boolean} we generalize \Cref{thm:Toffupper} to Boolean functions and prove \Cref{thm:fourier1norm}. 
Then in \Cref{sec:pdt} we describe the relationship between $T$-count and randomized parity decision tree complexity and prove \Cref{thm:Toffunitarylower}, \Cref{thm:PDT-thm-intro}, and \Cref{thm:adptive-lower-PDT}. Finally, in \Cref{sec:examples}, we justify the bounds in \Cref{tab:function_costs}.

\section{Clifford+\texorpdfstring{$T$}{T} circuits}\label{sec:CliffordT}

To define our models precisely, we need to discuss some distance measures on quantum states, unitaries, and channels.

For any two mixed states $\rho$ and $\sigma$, let the trace distance between them be denoted by $D(\rho, \sigma)=\frac{1}{2}\|\rho-\sigma\|_1$, where $\|A\|_1=\Tr(\sqrt{A^\dagger A})$. For ease of notation, we also use $D(\ket{\psi},\ket{\phi})$ to mean $D(\ketbra{\psi}{\psi},\ketbra{\phi}{\phi}).$ 
The trace distance has an operational interpretation: By the Holevo--Helstrom theorem~\cite[Theorem 3.4]{Wat18}, the maximum success probability of distinguishing the two states given one copy is $\frac12 + \frac12 D(\rho,\sigma)$. 
In particular, two states have trace distance $0$ if they are identical and trace distance $1$ if they are orthogonal (and hence perfectly distinguishable).

For any two mixed states $\rho$ and $\sigma$, let the fidelity between them be denoted by $F(\rho,\sigma) = \big(\|\sqrt{\rho}\sqrt{\sigma}\|_1\big)^2$. 
When one of the states is pure, we get the simpler formula $F(\rho,\ketbra{\psi}{\psi}) = \braket{\psi|\rho|\psi}$.
When the trace distance between two states is close to $0$ the fidelity is close to $1$, and vice versa. This is quantified by the Fuchs-van de Graaf inequalities~\cite{FvdG99}:
\begin{equation}
1-\sqrt{F(\rho,\sigma)}\leq D(\rho,\sigma)\leq \sqrt{1-F(\rho,\sigma)}.
\label{eq:fvdf}
\end{equation}

For quantum channels or unitaries, the distance measure analogous to trace distance is the diamond distance. 

\begin{definition}[Diamond distance]
Let $\mathcal{E}_1, \mathcal{E}_2$ be quantum channels which map $n$-qubit states to $n$-qubit states. Let $\id_\ell$ denote the identity channel on a Hilbert space of $\ell$ qubits. Then
\begin{equation}
\Dd(\mathcal{E}_1,\mathcal{E}_2)
=\mathrm{sup}\left\{D\bigl(\mathcal{E}_1\otimes \id_{\ell} (\rho),\mathcal{E}_2\otimes \id_\ell(\rho)\bigr): \ell<\infty\right\}
\label{eq:sup}
\end{equation}
where the supremum is over $\ell\in \N$ and density matrices $\rho$ on $n+\ell$ qubits. 
\end{definition}

A consequence of this definition is that if two channels are $\eps$-close in diamond distance and both act on the same state $\rho$, then their output states are also $\eps$-close in trace distance.

We often encounter the diamond distance between a quantum channel $\mathcal{E}$ and a unitary channel $\Phi_U(\rho)=U\rho U^{\dagger}$, where $U$ is an $n$-qubit unitary. In this situation for ease of notation we write
\begin{equation}
\Dd(\mathcal{E},U)\equiv \Dd(\mathcal{E},\Phi_U).
\end{equation}

We will also use the fact that the supremum in \Cref{eq:sup} is achieved by a pure state. 
\begin{fact}[\cite{rosgen2005hardness}, Lemma 2.4]
Let $\mathcal{E}_1, \mathcal{E}_2$ be quantum channels which map $n$-qubit states to $n$-qubit states. Then
\begin{equation}   
\Dd(\mathcal{E}_1,\mathcal{E}_2)
=\mathrm{sup}\left\{D\bigl(\mathcal{E}_1\otimes \id_{\ell} (|\psi\rangle\langle \psi|),\mathcal{E}_2\otimes \id_\ell(|\psi\rangle\langle \psi|)\bigr): \ell<\infty\right\}
\label{eq:pure}
\end{equation}
where the supremum is over $\ell\in \N$ and $n+\ell$-qubit pure states $|\psi\rangle$.
\label{fact:pure}
\end{fact}
It is also known that the supremum in \Cref{eq:sup,eq:pure} is achieved by a finite value of $\ell$ but we will not need this fact.

We now define the three models of Clifford+$T$ circuits discussed in the introduction. A unitary Clifford+$T$ circuit is the standard circuit that comes to mind when thinking of a Clifford+$T$ circuit. 

\begin{definition}[Unitary Clifford+$T$ circuit]\label{def:unitaryT} 
Let $V$ be a quantum circuit composed of Clifford gates and $T$ gates that acts on an $n$-qubit input state $|\psi\rangle$ along with some ancilla qubits initialized in the all-zeros state. The $\epsilon$-approximate unitary $T$-count of $U$, denoted 
$\Tunit_\eps(U)$
is the minimum number of $T$ gates in any such circuit $V$ that satisfies $\Dd(\Tr_{\mathrm{anc}}[\Phi_V],U) \leq \epsilon$, where $\mathrm{anc}$ denotes the ancilla register.
\end{definition}

As discussed, if two channels are $\eps$-close in diamond distance, then replacing one by the other in a quantum circuit at most changes the output state by at most $\eps$ in trace distance. 
In particular, this means if a unitary $U$ is used $k$ times in a quantum circuit and is replaced by a channel $\Phi$ that is $\eps/k$-close in diamond distance, then the output state of the resulting quantum state is at most $\eps$ close to the original output state in trace distance.

We now define mixed Clifford+$T$ circuits, which is the model in which we establish all our algorithmic results.

\begin{definition}[Mixed Clifford+$T$ circuit] \label{def:mixedT} 
Consider a probability distribution $\{p_i\}_i$ over unitary Clifford+$T$ circuits $V_i$ each of which act on an $n$-qubit input state along with an ancilla register consisting of $a$ qubits initialized in the all-zeros state. Let $k$ denote the maximum number of $T$ gates used by any one of the Clifford+$T$ circuits $V_i$. Define an associated $n$-qubit quantum channel  $\mathcal{E}(\rho)=\mathrm{Tr}_{\mathrm{anc}}\left[\sum_{i} p_i V_i(\rho\otimes |0^a\rangle \langle 0^a| )V_i^{\dagger}\right]$. The $\epsilon$-approximate mixed unitary $T$-count of $U$, denoted
$\Tmix_\eps(U)$
is the minimum $k$ of any such channel satisfying $\Dd(\mathcal{E},U)\leq \epsilon$.
\end{definition}

As discussed, implementing a mixed Clifford+$T$ circuit on a quantum computer requires no additional quantum hardware from the quantum computer since the probabilistic sampling can be done by the classical compiler.

The last and most powerful model is the adaptive Clifford+$T$ circuit. These circuits are also sometimes called ``circuits with measurement and classical feed-forward'' in the literature.

\begin{definition}[Adaptive Clifford+$T$ circuit]\label{def:adaptiveT} 
Consider a circuit that begins with an input state $|\psi\rangle$ as well as some ancilla qubits initialized in the all-zeros state, then applies a sequence of gates and single-qubit measurements in the computational basis. Each of the gates is either a $T$ gate or a Clifford gate, and may depend (deterministically or probabilistically) on the measurement outcomes that have been observed so far. At the end of the computation we discard the ancilla qubits, so the adaptive Clifford+$T$ circuit defines a channel $\mathcal{E}$ that maps $n$-qubit states to $n$-qubit states. Let $k(\psi)$ denote the expected number of $T$ gates used by the adaptive Clifford+$T$ circuit (over measurement outcomes and realizations of the randomness used) on input $|\psi\rangle$, and let $k=\mathrm{sup} _{|\psi\rangle} k(\psi)$.  The $\epsilon$-approximate adaptive $T$-count of $U$, denoted
$\Tadapt_\eps(U)$
is the minimum $k$ of any such channel satisfying $\Dd(\mathcal{E},U)\leq \epsilon$.
\end{definition}

We do not use the power of this adaptive Clifford+$T$ circuit model in any of the algorithms in this paper. We introduce the stronger model only to highlight the difference with our model, and because some of our lower bounds will hold even in the stronger model. See \cite{GKW24} for a more detailed discussion of this model.

\newpage

\section{Multi-qubit Toffoli\label{sec:mct}}

In this Section we prove \Cref{thm:Toffupper} and \Cref{thm:Toffadaptivelower}. 

\subsection{Algorithm}

The approximate implementation of the multi-qubit Toffoli gate that we use to establish \Cref{thm:Toffupper} is presented as \Cref{algo:OR_n-channel}. 
\begin{algorithm}[htbp!]
\caption{Approximate implementation of $\Toff_n$}
\label{algo:OR_n-channel}
\makeatletter
\makeatother
\SetKwInOut{Input}{Input}
\SetKwInOut{Parameter}{Parameter}
\SetKwInOut{Output}{Output}
\textbf{Input: }A positive integer $k$.

\For{$j\gets 1$ \KwTo $k$}{
  Sample a uniformly random subset $S_j\subseteq[n-1]$\label{lin:OR_random_subsets}.
}

Define a Boolean function $g\colon\{0,1\}^{n-1}\to\{0,1\}$ by\label{lin:OR_gx_def}
\begin{equation*}
g(x)\equiv \OR_k\big(\XOR_{S_1}(x),\XOR_{S_2}(x),\ldots,\XOR_{S_k}(x)\big).
\end{equation*}

Implement the unitary $W_g=X^{\otimes n-1} U_g X^{\otimes n}$, where $U_g$ is defined in \Cref{eq:Uf}.
\end{algorithm}

\begin{theorem}\label{thm:toffdiamond}
The mixed Clifford+$T$ circuit from \Cref{algo:OR_n-channel} defines a quantum channel \begin{equation}
\mathcal{E}(\rho)=\mathbb{E}_{S_1,\ldots,S_k} \left[W_g\rho W_g^{\dagger}\right]~~\text{satisfying}~~\Dd(\mathcal{E}, \Toff_n)\leq \frac{4}{2^{k}}.
\end{equation}
\end{theorem}

\begin{proof}
Let $E_g \equiv \{x\in \{0,1\}^{n-1}\, |\, g(x) \neq \OR_{n-1}(x)\}$ be the set of inputs on which $g$ is incorrect for a given choice of sets $S_1,\ldots,S_k$. As discussed in the Introduction (see \Cref{eq:probbound}), for any $x\in \{0,1\}^{n-1}$, we have $\Pr_{S_1,\ldots,S_k}[x \in E_g] \leq 2^{-k}$. For ease of notation, let $\eps_k\equiv 2^{-k}$.

Let $\ell\in \N$ and let $|\psi\rangle$ be any $n+\ell$ qubit pure state input:
\begin{equation}
|\psi\rangle=\sum_{x\in \{0,1\}^{n-1}}\sum_{y\in \{0,1\}}\sum_{z\in \{0,1\}^{\ell}}\alpha_{xyz}|x\rangle|y\rangle|z\rangle.
\end{equation}
Note that for all $x\notin E_g$, $y\in \{0,1\}$, and $z \in \B^\ell$, we have
\begin{equation}
(W_g \otimes \id_\ell) \ket{x,y,z} = (\Toff_n \otimes \id_\ell)\ket{x,y,z}.
\end{equation}
Denote $\Delta_g\equiv W_g-\Toff_n$. Then
\begin{align}
&D\Bigl((\mathcal{E}\otimes \id_\ell)(|\psi\rangle\langle \psi|), (\Toff_n\otimes \id_{\ell}) |\psi\rangle\langle \psi| (\Toff_n\otimes \id_{\ell}) \Bigr)\nonumber\\
&\quad=\frac{1}{2}\big\|\E[(W_g\otimes I_\ell)\ket\psi\!\bra\psi(W_g\otimes I_\ell)]-(\Toff_n\otimes I_\ell)\ket\psi\!\bra\psi(\Toff_n\otimes I_\ell)\big\|_1\\
&\quad\leq \frac12\left\|(\E[\Delta_g]\otimes I_\ell)\ket\psi\!\bra\psi(\Toff_n\otimes I_\ell)\right\|_1
+\frac12\left\|(\Toff_n\otimes I_\ell)\ket\psi\!\bra\psi(\E[\Delta_g]\otimes I_\ell)\right\|_1\nonumber\\
&\quad\qquad+\frac{1}{2}\,\E\left\|(\Delta_g\otimes I_\ell)\ket\psi\!\bra\psi(\Delta_g\otimes I_\ell)\right\|_1\label{eq:1-norm-decomposition},
\end{align}
given that
\begin{align}
\begin{aligned}
&(W_g\otimes I_\ell)\ket\psi\!\bra\psi(W_g\otimes I_\ell)-(\Toff_n\otimes I_\ell)\ket\psi\!\bra\psi(\Toff_n\otimes I_\ell)\\
&\qquad=(\Delta_g\otimes I_\ell)\ket\psi\!\bra\psi(\Toff_n\otimes I_\ell)
+(\Toff_n\otimes I_\ell)\ket\psi\!\bra\psi(\Delta_g\otimes I_\ell)\\
&\qquad\qquad+(\Delta_g\otimes I_\ell)\ket\psi\!\bra\psi(\Delta_g\otimes I_\ell)
\end{aligned}
\end{align}
for any $g$. Since for every $d\geq 1$ and any two vectors $u,v\in \mathbb C^d$ we have $\|uv^\top\|_1=\sum_{i=1}^d|u_iv_i|\leq \|u\|\!\cdot\!\|v\|$ by Cauchy--Schwartz, the first two terms in \Cref{eq:1-norm-decomposition} both can be upper bounded by
\begin{align}
\frac12\|\E[\Delta_g]\otimes I_\ell\ket\psi\|\cdot \|\Toff_n\otimes I_\ell\ket\psi\|
&=\frac12\|\E[\Delta_g]\otimes I_\ell\ket\psi\|\leq\frac12\|\E[\Delta_g]\otimes I_\ell\|=\epsilon_k
\end{align}
since for any $x\in\B^{n-1},y\in\B,z\in\B^{\ell}$ we have
\begin{align}
(\E[\Delta_g]\otimes\id_\ell)\ket{x,y,z}&=\Pr[g(x\oplus 1^{n-1})\neq \OR(x\oplus 1^{n-1})](\ket{x,y\oplus 1,z}-\ket{x,y,z})\\
&=\begin{cases}
0,& x=1^{n-1}\\
\epsilon_k(\ket{x,y\oplus 1,z}-\ket{x,y,z}),&\text{otherwise}.
\end{cases}
\end{align}
As for the third term in \Cref{eq:1-norm-decomposition}, we have 
\begin{align}
\frac{1}{2}\,\E\left\|(\Delta_g\otimes I_\ell)\ket\psi\!\bra\psi(\Delta_g\otimes I_\ell)\right\|_1
&=\frac12\,\E\|(\Delta_g\otimes I_\ell)\ket\psi\|^2\\
&\leq\frac12\,\E\bigg[\sum_{(x\oplus 1^{n-1)}\in E_g}\sum_{yz\in \{0,1\}^{\ell+1}}4|\alpha_{x,y,z}|^2\bigg]\\
&= 2\mathrm{Pr}\left[x\oplus 1^{n-1}\in E_g\right]\cdot \sum_{xyz\in \{0,1\}^{n+\ell}}
 |\alpha_{xyz}|^2  
\\
&\leq 2\epsilon_k,
\end{align}
where we used the facts that $\Pr[x\oplus 1^{n-1} \in E_g] \leq \eps_k$ for all $x$ and $\sum_{xyz}|\alpha_{xyz}|^2=1$. This gives \begin{equation}
D\Bigl((\mathcal{E}\otimes \id_\ell)(|\psi\rangle\langle \psi|), (\Toff_n\otimes \id_{\ell}) |\psi\rangle\langle \psi| (\Toff_n\otimes \id_{\ell}) \Bigr)\leq 
4\eps_k=2^{2-k}.
\end{equation}
Finally, since $|\psi\rangle$ is an arbitrary pure state on $n+\ell$ qubits we can use \Cref{fact:pure} to conclude $\Dd(\mathcal{E}, \Toff_n)\leq 2^{2-k}$.
\end{proof}

We are now ready to prove \Cref{thm:Toffupper}, which we restate:
\Toffupper*

\begin{proof}
We approximate $\Toff_n$ using \Cref{algo:OR_n-channel} with the choice $k=\lceil \log(1/\epsilon)\rceil+2$. From \Cref{thm:toffdiamond} this ensures $\Dd(\mathcal{E}, \Toff_n)\leq \epsilon$. 

Now let us consider the number of $T$ gates needed to implement the unitary $W_g$ in line 4 of the algorithm. First we need to reversibly compute each of the parities $\XOR_{S_j}(x)$. This can be done using a sequence of CNOT gates, each of which is Clifford. Clearly the Pauli gates $X^{\otimes n}$ are also Clifford, so the only non-Clifford operation is the reversible computation of $\OR_k$, which as we have discussed is Clifford-equivalent to $\Toff_{k+1}$. Thus the $T$-count of the mixed Clifford+$T$ circuit that approximates $\Toff_n$ to within $\epsilon$ diamond-distance error is at most the unitary $T$-count of exactly implementing $\Toff_{k+1} = \Toff_{\lceil \log(1/\epsilon)\rceil+3}$. (It would also be fine to have a mixed Clifford+$T$ implementing $\Toff_{k+1}$ here, but the error would have to be very small, of the order of $1/2^k$, at which point a unitary implementation is just as efficient as shown in \Cref{thm:Toffadaptivelower}.)
\end{proof}

\subsection{Lower bound}

We now prove this algorithm is optimal (up to constants). As in Beverland et al.~\cite{beverland2020lower}, we establish this using the stabilizer nullity proof technique. 
The stabilizer nullity is a function $\nu (\cdot)$ defined on all $n$-qubit quantum states as follows:
\begin{equation}
\nu(\sigma)=n-\log \left(\left|\{P\in \{\pm1\} \cdot \{I,X,Y,Z\}^n: P\sigma =\sigma\}\right|\right).
\end{equation}

The stabilizer nullity is one way to quantify magic for quantum states; it has the following properties:
\begin{enumerate}[itemsep=0pt]
\item{$\nu(|\phi\rangle\langle \phi|) \in \{0,\ldots, n\}$, with $\nu(|\phi\rangle\langle \phi|)=0$ if and only if $|\phi\rangle$ is a stabilizer state.} 
\item{$\nu(C\rho C^{\dagger}) = \nu(\rho)$ whenever $C$ is a Clifford unitary. }
\item{$\nu(T_j\rho T_j^{\dagger})\leq \nu(\rho)+1$ where $T_j$ is the single-qubit $T$ gate acting on qubit $j\in [n]$.}
\item{$\nu(\rho')\leq \nu(\rho)$, where 
\begin{equation}
\rho'=\frac{1}{\mathrm{Tr}(\rho(\id+P)/2)}\left(\frac{\id+P}{2}\right)\rho\left(\frac{\id+P}{2}\right), \qquad \textbf{(Pauli postselection)}
\label{eq:pps}
\end{equation}
is the state obtained by measuring a Pauli $P$ and postselecting on the $+1$ outcome  (assuming this state is well defined, i.e.,  $\mathrm{Tr}(\rho(\id+P)/2)\neq 0$). }
\item{$\nu(\rho\otimes \sigma)=\nu(\rho)+\nu(\sigma)
$}
\end{enumerate}

Properties 1, 2, and 5 follow straightforwardly from the definition, see \cite[Proposition 2.3]{beverland2020lower} for a proof of property 4.\footnote{Although Beverland et al. only state this Proposition for pure states (as they only define stabilizer nullity for pure states), the proof of Prop 2.3 given in Ref. \cite{beverland2020lower} extends straightforwardly to mixed states.} For property 3, note that the single qubit magic state $|T\rangle=\frac{1}{\sqrt{2}}(|0\rangle+e^{-i\pi/4}|1\rangle)$ has nullity 
\begin{equation}
\nu(|T\rangle\langle T|)=1,
\end{equation}
and that we can implement a $T$ gate by adjoining a magic state (increasing nullity by $1$) and then performing a sequence of nullity non-increasing operations:
\begin{equation}
2|0\rangle\langle 0|_B\mathrm{CNOT}_{jB}\left(\rho\otimes |T\rangle\langle T|_B\right) \mathrm{CNOT}_{jB}|0\rangle\langle 0|_B= T_j\rho T_j^{\dagger}\otimes |0\rangle\langle 0|_B.
\label{eq:tgadget}
\end{equation}
Using property 5, we see that the nullity of the RHS, which is at most $\nu(\rho)+1$ is equal to that of $\nu(T_j\rho T_j^{\dagger})$. We note that this argument generalizes (replacing $\pi/4\leftarrow\theta$ everywhere) to show that stabilizer nullity can only increase by at most one if we apply any single-qubit diagonal unitary $D=\mathrm{diag}(1,e^{i\theta})$.

We first establish that all states in a ball of radius $2/2^n$ around $C^{n-1}Z|+\rangle^{\otimes n}$ have maximal stabilizer nullity.

\begin{lemma}\label{lem:epsnullity}
Let $n\geq 3$ and $|\Phi\rangle\equiv C^{n-1}Z|+\rangle^{\otimes n}$. Suppose $\omega$ is an $n$-qubit state such that
$D(\omega, |\Phi\rangle\langle \Phi|)< 2/2^n$.
Then $\nu(\omega)=n$.
\end{lemma}
\begin{proof}
By directly computing all Pauli expected values in the state $|\Phi\rangle$ (see~\cite[Proposition 4.2]{beverland2020lower}), for $n\geq 3$ we have
\begin{equation}\label{eq:maxP}
\max_{P\in \{I,X,Y,Z\}^{\otimes n}: P\neq I} |\langle \Phi|P|\Phi\rangle| =1-\frac{4}{2^n}.
\end{equation}

Toward a contradiction, assume $\nu(\omega)<n$. Then $\omega$ has a nontrivial stabilizer $P$ satisfying $P\omega=\omega$. Now consider the two-outcome measurement $\{\Pi,\id-\Pi\}$, where $\Pi=\frac{\id+P}{2}$. On performing this measurement on $\omega$, since $P\omega=\omega$, the probability vector corresponding to the two outputs is $(1,0)$, since we always get the first outcome. On the other hand, performing this measurement on $\ket{\Phi}$ has the following probability of getting the first outcome:
\begin{equation}
    \Tr(\Pi \ketbra{\Phi}{\Phi}) = \frac{1}{2}(1 + \braket{\Phi|P|\Phi}) \leq \frac{1}{2}\Big(1 + 1 - \frac{4}{2^n}\Big) =1-\frac{2}{2^n}.
\end{equation}
Thus the resulting two-outcome probability distribution is $(p,1-p)$ for $p\leq 1-2/2^n$. The total variation distance between $(p,1-p)$ and $(1,0)$ is $1-p\geq 2/2^n$. Since the total variation distance after measurement is upper bounded by the trace distance before measurement~\cite[Theorem 9.1]{NC10}, we must have $D(\omega, |\Phi\rangle\langle \Phi|)\geq 2/2^n$.
\end{proof}

If we have an adaptive Clifford+$T$ circuit that implements $\Toff_n$ to within error $\epsilon$, we can use it to prepare an $\eps$-approximation $|\Psi\rangle$ to the state $C^{n-1}Z|+\rangle^{\otimes}$ (since $\Toff_n$ is Clifford equivalent to $C^{n-1}Z$). In \Cref{thm:Toffepslower} we first focus our attention on the case where $\epsilon$ is exponentially small in $n$. Then we can use \Cref{lem:epsnullity} to infer that $|\Psi\rangle$ has stabilizer nullity $\nu(\Psi)=n$. To prove the theorem we then show that the expected number of $T$ gates used by the adaptive Clifford+$T$ circuit upper bounds $\nu(\Psi)/2$. In order to show this we use the following proposition which relates adaptive Clifford+$T$ circuits to Clifford circuits with Pauli postselection.

\begin{proposition}
[\protect{\cite[Claim 4.5]{GKW24}}]\label{prop:adaptivetopostselected}
Suppose an adaptive Clifford+$T$ circuit acting on the input state $|0^n\rangle$ prepares an $n$-qubit output state  $\ket{\Phi}$ to within trace distance $\epsilon$, and uses $t$ $T$ gates in expectation. Then there is a Clifford circuit with Pauli postselections $C$, such that
\begin{equation}
C(|0^n\rangle |T\rangle^{\otimes 2t}|0^a\rangle)=|\phi\rangle|0^{2t+a}\rangle
\end{equation}
for some $n$-qubit state $\ket{\phi}$ satisfying $D(|\phi\rangle, |\Phi\rangle)\leq \sqrt{6\eps}$.
\end{proposition}

\begin{theorem}\label{thm:Toffepslower}
    Let $n\geq 3$ and $\eps \leq 1/4^{n+1}$. Then $\Tadapt_\eps(\Toff_n) \geq n/2$.
\end{theorem}
\begin{proof}
Let $n\geq 3$ and $\epsilon\leq 1/4^{n+1}$ be given. Consider an adaptive Clifford+$T$ circuit that $\epsilon$-approximately implements $\Toff_n$ and such that the expected number of $T$ gates used by the circuit on the worst-case input state\footnote{i.e. the input state where this expected number of $T$ gates is maximal} is $\Tadapt_\eps(\Toff_n)$. Such a circuit always exists by definition of $\Tadapt_\eps$. The expected number of $T$ gates used by the circuit starting from input state $|+\rangle^{\otimes n}$ is $t\leq \Tadapt_\eps(\Toff_n)$.  Below we show that $t \geq n/2$.

Let $|\Phi\rangle\equiv C^{n-1}Z|+\rangle^{\otimes n}$. Since $\Toff_n$ is Clifford-equivalent to $C^{n-1}Z$, and since $|+\rangle^{\otimes n}$ is a stabilizer state, by adding some Clifford gates to our adaptive circuit we get an adaptive Clifford+$T$ circuit that starts with $|0^n\rangle$ and prepares $\ket\Phi$ to within error $\eps$ using the same expected number of $T$ gates $t$. Applying \Cref{prop:adaptivetopostselected}, we infer that there exists a Clifford circuit with Pauli postselections $C$ such that
\begin{equation}
C(|0^n\rangle|T\rangle^{\otimes 2t}|0^a\rangle)=|\Psi\rangle|0^{2t+a}\rangle
\end{equation}
for some state $\ket\Psi$ satisfying $D(\ket\Psi,\ket\Phi) \leq \sqrt{6\eps} \leq \sqrt{3/2} \cdot 2^{-n} < 2/2^n$.

The stabilizer nullity of the input state
\begin{equation}
|0^n\rangle|T\rangle^{\otimes 2t}|0^a\rangle=|0^n\rangle T^{\otimes 2t}|+\rangle^{2t}|0^a\rangle
\end{equation}
is at most $2t$ since $|0\rangle$ and $|+\rangle$ are stabilizer states and each $T$ gate can increase the nullity by at most $1$. Stabilizer nullity does not increase under Cliffords or Pauli postselections, so the output state of $C$ also has stabilizer nullity upper bounded by $2t$:
\begin{equation}
\nu (|\Psi\rangle|0^{2t+a}\rangle)=\nu(|\Psi\rangle)+\nu(|0^{2t+a}\rangle)=\nu(|\Psi\rangle)\leq 2t.
\end{equation}
Lastly, from \Cref{lem:epsnullity} we know that since $D(\ket\Psi,\ket\Phi)<2/2^n$, we must have $\nu(\ket\Psi)\geq n$, which gives $2t\geq n$.
\end{proof}

We are now ready to prove our lower bound, which we restate for convenience.
\Toffadaptivelower*

\begin{proof}
First suppose that $\eps \leq 1/4^{n+1}$. Then \Cref{thm:Toffepslower} gives a lower bound
\begin{equation}
\Tadapt_\eps(\Toff_n) \geq n/2.
\end{equation}
On the other hand if $4^{n+1} \geq 1/\eps$ then let $n'<n$ be the largest integer satisfying $4^{n'+1} \leq 1/\epsilon$. Note that $n'=\Theta(\log(1/\epsilon))$. Then since any circuit for $\Toff_n$ can also implement $\Toff_{n'}$,
\begin{equation}
\Tadapt_\eps(\Toff_n) \geq \Tadapt_\eps(\Toff_{n'}) \geq n'/2 = \Omega(\log(1/\epsilon)).
\end{equation}
In both cases we have shown $\Tadapt_\eps(\Toff_n) \geq \Omega(\min\{n,\log(1/\epsilon)\}).$
\end{proof}

\section{Generalization\label{sec:boolean}}

In this section, we generalize \Cref{thm:Toffupper} to upper bound the $T$-count of other Boolean functions and establish \Cref{thm:fourier1norm}. 

For any $S\subseteq [n]$ we write $\widehat{f}(S)$ for the Fourier coefficient of $f$ at $S$, defined by \Cref{eq:fourier}, and we write  $\|\widehat f\|_1\equiv \sum_{S\subseteq [n]}|\widehat f(S)|$ for the Fourier $1$-norm of $f$.

Inspired by a sampling procedure introduced by Grolmusz~\cite{Gro97} (see also \cite[Lemma 7]{bhrushundi2014property} for a proof in the context of randomized parity decision trees), in \Cref{algo:Uf-channel} we construct a mixed Clifford+$T$ circuit that approximates $U_f$ (defined in \Cref{eq:Uf}) in the sense described below.

\begin{algorithm}[htbp!]
\caption{Approximate implementation of $U_f$}
\label{algo:Uf-channel}
\SetKwInOut{Parameters}{Parameters}
\SetKwInOut{Input}{Input}
\textbf{Input: }A Boolean function $f:\B^n\to\B$ and a positive integer $k$.

\For{$j\gets 1$ \KwTo $k$}{
  Sample $S_j\subseteq[n]$ independently from distribution $p(S)=|\widehat f(S)|/\|\widehat f\|_1$\label{lin:Fourier_random_subsets}.
}
Define a Boolean function $g\colon\{0,1\}^n\to\{0,1\}$ by rounding the sum\label{lin:Fourier_gx_def}
\begin{equation}
g(x)=
\begin{cases}
1 & \text{if }~~\displaystyle \frac{\|\widehat f\|_1}{k} \sum_{i=1}^k \sign(\widehat{f}(S_i))\,\chi_{S_i}(x)\ge \dfrac{1}{2},\\
0 & \text{otherwise.}
\end{cases}
\end{equation}

Implement the unitary $U_g$ (defined in \Cref{eq:Uf}). \label{lin:Fourier_Ug}
\end{algorithm}
\begin{theorem}
The mixed Clifford+$T$ circuit from \Cref{algo:Uf-channel} defines a quantum channel 
\begin{equation}
\mathcal{E}(\rho)=\mathbb{E}_{S_1,\ldots,S_k} \left[U_g\rho U_g^{\dagger}\right] \quad \text{satisfying} \quad \Dd(\mathcal E, U_f)\leq 8\exp\left(-k/(8\|\widehat{f}\|_1^2)\right).
\end{equation}
\label{thm:uf}
\end{theorem}
 
Assuming this, let us show that \Cref{thm:fourier1norm} (restated here) follows:
\FourierOneNorm*

\begin{proof}
Set $k=\big\lceil 8\|\widehat{f}\|_1^2 \ln(8/\epsilon) \big\rceil = O(\|\widehat{f}\|_1^2 \log(1/\epsilon))$. Plugging this into \Cref{thm:uf} we see that the mixed Clifford+$T$ circuit implements a channel $\mathcal{E}$ satisfying $\Dd(\mathcal E, U_f)\leq \epsilon$.

As for the $T$-count, note that implementing $U_g$ involves first reversibly computing $k$ parity functions $\chi_{S_i}(x)$ which can be done using a sequence of CNOT gates which are Clifford. We then need to coherently compute the sum $\sum_{i=1}^k \sign(\widehat{f}(S_i)) \cdot \chi_{S_i}(x)$ and compare it to the threshold. This requires implementing a $k$-input threshold function, which can be implemented with a $T$-count of $O(k)$. This follows since even classical circuits can implement any symmetric Boolean function with linear $\AND$-count, the classical analogue of $T$-count~\cite{BPP00}.
The total $T$-count of the implementation is $O(k)=O(\|\widehat{f}\|_1^2 \cdot \log(1/\epsilon))$.
\end{proof}

We shall use the following Lemma in the proof of \Cref{thm:uf}. It states that, for any fixed input $x$, the function $g(x)$ from \Cref{algo:Uf-channel} equals $f(x)$ with high probability.

\begin{lemma}\label{lem:fourier-classical-error}
Let $g \colon \{0,1\}^n \to \{0,1\}$ be the Boolean function defined in \Cref{lin:Fourier_gx_def} of \Cref{algo:Uf-channel}. For any $x \in \{0,1\}^n$, we have
\begin{equation}
\Pr[g(x) \neq f(x)] \leq 2\exp\!\left(-\frac{k}{8\|\widehat{f}\|_1^2}\right),
\end{equation}
where the probability is over the random subsets $S_1,\ldots,S_k$ sampled in \Cref{lin:Fourier_random_subsets}.
\end{lemma}

\begin{proof}
We can rewrite $g(x)$ as $g(x)=\lfloor\tilde{g}(x)\rceil$, where $\lfloor y \rceil$ is the nearest integer to $y$, and  
\begin{equation}
    \tilde{g}(x) = \frac{\|\widehat f\|_1}{k} \sum_{i=1}^k \sign(\widehat{f}(S_i))\,\chi_{S_i}(x).
\end{equation}
Thus $\Pr[g(x)\neq f(x)] \leq \Pr[|f(x)-\tilde{g}(x)|\geq 1/2]$. The random variable $\tilde{g}(x)$ is the sum of $k$ independent and identically distributed random variables, which we call $X_i \equiv \frac{\|\widehat f\|_1}{k} \sign(\widehat{f}(S_i)) \cdot \chi_{S_i}(x)$ for $i\in[k]$. 
Since $S_1,\ldots,S_k$ are sampled independently, the $X_i$'s are independent. Furthermore, each $X_i$ is bounded in the interval $\left[-\frac{\|\widehat f\|_1}{k}, +\frac{\|\widehat f\|_1}{k}\right]$. Note that $\mathbb{E}[X_i] = f(x)/k$, and since $\tilde{g}(x) = \sum_{i=1}^k X_i$, By the linearity of expectation, $\mathbb{E}[\tilde{g}(x)] = f(x)$. 

Hoeffding's inequality says for the sum of $k$ independent random variables $X_i$ in the range $[-R,+R]$, we have
$\Pr[\bigl|\sum_iX_i - \E\bigl[\sum_i X_i\bigr]\bigr|\geq t]\leq 2\exp(-\frac{t^2}{2R^2}).$  
Applying this to $\tilde{g}(x)$, we get
\begin{equation}
    \Pr[g(x)\neq f(x)] \leq \Pr[|f(x)-\tilde{g}(x)| \geq \frac12] 
    \leq 2 \exp(-\frac{k}{8\|\widehat f\|_1^2}).\qedhere
\end{equation}
\end{proof}

With this Lemma in hand, the proof of \Cref{thm:uf} follows that of \Cref{thm:toffdiamond}.
\begin{proof}[Proof of \Cref{thm:uf}]
Let $\epsilon_k \equiv 2\exp(-k/(8\|\widehat{f}\|_1^2))$ and $E_g \coloneqq \{x\in \{0,1\}^n\, |\, g(x) \neq f(x)\}$. By \Cref{lem:fourier-classical-error}, for any fixed $x\in \{0,1\}^n$ we have $\Pr[x \in E_g] \leq \epsilon_k$. 

Now we are in the same situation as \Cref{thm:toffdiamond}, but with a different value for $\eps_k$. The entire proof goes through and we reach the conclusion that 
$\Dd(\mathcal{E}, U_f)\leq 4\epsilon_k$.
\end{proof}

\section{Randomized parity decision trees\label{sec:pdt}}

In this section, we present lower bounds on the $T$ count of Boolean functions using a complexity measure known as \emph{non-adaptive parity decision tree complexity}. 

Parity decision trees were first introduced by Kushilevitz and Mansour~\cite{KM93}, generalizing standard decision trees.
Given access to an $n$-bit string $x\in\B^n$, a standard decision tree queries input bits $x_i$ at unit cost, whereas a parity decision tree can query any parity function $\XOR_S(x)$ for an $S$ of its choice at unit cost. In this work we only use the concept of a non-adaptive decision tree, in which the set of parity queries is fixed in advance, and the output depends only on the collection of their values.

\begin{definition}[Non-adaptive parity decision tree]\label{def:NAPDT}
A non-adaptive (deterministic) parity decision tree with depth $k$ is a fixed collection of subsets $S_1,\ldots,S_k \subseteq [n]$ together with a deterministic function $g:\B^k \to \B$. It is said to compute the Boolean function $g(\XOR_{S_1}(x),\ldots,\XOR_{S_k}(x))$.

The non-adaptive parity decision tree complexity of a Boolean function $f$, denoted $\NAPDT(f)$, is the minimum depth among all parity decision trees that compute $f$ correctly on every input.
\end{definition}

\Cref{def:NAPDT} extends to the randomized setting in the standard way, by allowing a probability distribution over parity decision trees.

\begin{definition}[Non-adaptive randomized parity decision tree]
A non-adaptive randomized parity decision tree with depth $k$ is a probability distribution  over non-adaptive parity decision trees of depth at most $k$, and its output on an input $x\in\B^n$ is the distribution on $\B$ obtained by sampling a deterministic parity decision tree from this distribution and computing its output on $x$.

For any Boolean function $f$ and $\epsilon\geq 0$, the non-adaptive randomized parity decision tree complexity $\RNAPDT_{\epsilon}(f)$ is defined as the minimum depth of a randomized parity decision tree that outputs $f(x)$ with probability at least $1-\epsilon$ for all $x$.
\end{definition}

Equivalently, $\RPDTna(f)\leq k$ if there exists a probability distribution $p_i$ over Boolean functions $g_i$ such that for all $i$, $\PDTna(g_i)\leq k$ and 
\begin{equation}
    \Pr[f(x) \neq g_i(x)] \leq \eps~~\text{for all}~~x \in \B^n.
\end{equation}

In this section we establish \Cref{thm:PDT-thm-intro}, restated here for convenience:
\PDTintro*

One might also wonder if these bounds could be improved using the stronger and better studied model of adaptive parity decision trees. Unfortunately, even with only $1$ round of adaptivity and $1$ bit queried adaptively, which is the least adaptive an algorithm could be, there is an exponential separation between $\mathrm{PDT}(f)$ and even $\Tadapt_{1/3}(U_f)$. Let $f$ be the Index function on $k+2^k$ bits, defined as $f(x,y)=y_x$ for $x\in\B^k$ and $y\in\B^{2^k}$, where the first $k$ bits specify a position in the string of length $2^k$ and the goal is to output that bit. It is easy to see that an adaptive algorithm can first query $x$ and then $y_x$, which is 1 bit and uses only 1 round of adaptivity, giving a total $k+1$ bits queried. But the index function on $k+2^k$ bits includes as a sub-function every Boolean function on $k$ bits by fixing the $2^k$ bits to the be the truth table of the function under consideration. We know there exists a Boolean function on $g$ bits with $\Tadapt_{1/3}(U_g) = \Omega(2^{k/2})$~\cite{GKW24}, which implies the same lower bound for the Index function.

\subsection{PDT complexity lower bounds unitary \texorpdfstring{$T$}{T}-count}\label{sec:unitary-circuit-PDT}

In this subsection, we show that for any Boolean function $f\colon\{0,1\}^n\to\{0,1\}$, 
we can lower bound $\Tunitary_\epsilon(U_f)$ for any $\epsilon\in[0,1/2)$ by its non-adaptive parity decision tree complexity $\PDTna(f)$.

\begin{theorem}[Part 1 of \Cref{thm:PDT-thm-intro}]\label{prop:Tunitary-upper-bound-PDT}
For any Boolean function $f:\{0,1\}^n \to \{0,1\}$ and for any $\epsilon\in[0,1/2)$, we have
\begin{equation}
\Tunitary_\epsilon(U_f) \geq \PDTna(f) - 1.
\end{equation}
\end{theorem}

Since it is well-known that $\PDTna(\OR_n)=n$ by a simple adversary argument,\footnote{For any $n-1$ fixed parity queries, if all the parities evaluate to $0$, there exists at least one non-zero input $x$ consistent with this, and the PDT cannot distinguish $x$ from $0^n$.} 
and $\Toff_n = X^{\otimes n-1}U_{\OR_{n-1}}X^{\otimes n}$, we immediately obtain \Cref{thm:Toffunitarylower} (restated below) as a corollary.

\Toffunitarylower*

To prove \Cref{prop:Tunitary-upper-bound-PDT}, consider a  unitary Clifford+$T$ circuit containing at most $k$ $T$ gates that (approximately) computes a Boolean function $f(x)$. The input is a basis state $\ket{x}$ together with $a$ ancillas, and the output is obtained by measuring the first qubit in the $Z$ basis. We shall allow the ancilla register to be initialized an arbitrary $a$-qubit state that we denote $|\phi_{\mathrm{in}}\rangle$. We consider the probability of the measurement result being 1,
\begin{equation}\label{eqn:unitary-circuit-probability}
p^{\mathrm{output}}(x)\coloneqq(\bra{x}\otimes \bra{\phi_{\mathrm{in}}})\,U^\dagger \Pi\,U(\ket{x}\otimes \ket{\phi_{\mathrm{in}}}),
\end{equation}
where $\Pi=\tfrac{\id-Z_1}{2}$. We show that there exists a non-adaptive deterministic parity decision tree of depth $k$ that can exactly compute $p^{\mathrm{output}}(x)$  and hence can output $1$ if and only if $p^{\mathrm{output}}(x)>1/2$. 

It will be convenient to introduce the following notation. 
Recall that any  Pauli operator $P\in \pm \{I,X,Y,Z\}^{\otimes N}$ can be written as
\begin{equation}
P=\pm i^{v\cdot w} X(v) Z(w) \qquad v,w\in \{0,1\}^N 
\label{eq:paulis}
\end{equation}
where $X(v)=\prod_{j\in [N]} X_j^{v_j}$ and $Z(w)=\prod_{j\in [N]} Z_j^{w_j}$. We say that $X(v)$ and $Z(w)$ are the $X$-type part and the $Z$-type part of $P$, respectively.  A Pauli $P\in \{I,X,Y,Z\}^{\otimes N}$ is said to be $Z$-type (resp. $X$-type) if its $X$-type (resp. $Z$-type) part is the identity. A Pauli $P\in \{I,X,Y,Z\}^{\otimes N}$ is said to be $Z$-type (resp. $X$-type) on a subset $A\subseteq [n]$ of the qubits if its $X$-type (resp. $Z$-type) part acts as the identity on all qubits in $A$. 

For any $P\in \pm\{I,X,Y,Z\}^{\otimes N}$, let $R(P)$ be the following $N$-qubit unitary:
\begin{align}
R(P) := \exp\!\left(-\frac{i\pi}{8}\cdot P\right).
\end{align}

The following Fact gives a canonical form for Clifford+$T$ circuits that use $k$ $T$ gates.

\begin{fact}[See e.g. \cite{gosset2014algorithm}]\label{lem:pauli-rot-nf}
Let $N$ be a positive integer. Let $U$ be an $N$-qubit unitary Clifford+$T$ circuit which uses $k$ $T$ gates. There exists a global phase $\mathrm{e}^{i\phi}$, an $N$-qubit Clifford unitary $C_0$, and Paulis $P_1,\dots,P_k\in \pm\{I,X,Y,Z\}^{\otimes N}$ such that
\begin{align}
U = \mathrm{e}^{i\phi} C_0 \Bigg(\prod_{j=1}^{k} R(P_j)\Bigg).
\end{align}
\end{fact}
\Cref{lem:pauli-rot-nf} is proved by first writing each $T$ gate in the Clifford+$T$ circuit as $e^{i\pi/8} R(-Z_j)$ where $j\in [N]$ is the qubit the gate acts on, and then commuting all the Cliffords to the left, see \cite{gosset2014algorithm} for details.
\begin{lemma}\label{lem:poly-form}
There exist $k+1$ subsets $S_0,S_1,\ldots,S_k\subseteq [n]$ and a polynomial $h:\{0,1\}^{k+1}\rightarrow \mathbb{R}$ such that the probability $p^{\mathrm{output}}(x)$ defined in \Cref{eqn:unitary-circuit-probability} satisfies
\begin{equation}
p^{\mathrm{output}}(x)=h\big(\XOR_{S_0}(x),\XOR_{S_1}(x),\ldots,\XOR_{S_k}(x)\big).
\end{equation}
\end{lemma}
\begin{proof}
Using \Cref{lem:pauli-rot-nf} and the definition of $p^{\mathrm{output}}(x)$ gives $(n+a)$-qubit Paulis $P_1,P_2,\ldots, P_k$ such that
\begin{equation}\label{eqn:probility-product-direct-expansion}
p^{\mathrm{output}}(x) = \frac{1}{2} - \frac{1}{2} \bra{x}\otimes \langle \phi_{\mathrm{in}}| \bigg(\prod_j R(P_j)\bigg)^\dagger P_0 \bigg(\prod_j R(P_j)\bigg) \ket{x}\otimes |\phi_{\mathrm{in}}\rangle.
\end{equation}
where $P_0=C_0^{\dagger} Z_1 C_0$. Consider the group
\begin{equation}
\mathcal{W}\equiv\langle P_0, P_1,P_2,\ldots, P_k\rangle
\label{eq:w}
\end{equation}
and
\begin{equation}
\mathcal{W}_Z\equiv\{P\in \mathcal{W}: P \text{ is  $Z$-type on qubits }\{1,2,\ldots,n\}\}.
\end{equation}
Let $Z(b_0),Z(b_1),\ldots, Z(b_k)$ be the $Z$-type parts of $P_0,P_1,\ldots, P_k$ respectively. Here $b_0,b_1,\ldots, b_k\in \{0,1\}^{n+a}$. Let $\beta_j$ consist of the first $n$ bits of $b_j$, for each $j \in \{0,1,\ldots, n\}$. Then any Pauli $P\in \mathcal{W}_Z$ can be written as $P=P'\otimes Q$, where
\begin{equation}
P'\in \pm \langle Z(\beta_0),Z(\beta_1),\ldots, Z(\beta_k)\rangle\text{ and } Q\in \{I,X,Y,Z\}^{\otimes a}.
\label{eq:wz}
\end{equation}
Since $R(P_j)=\cos(\pi/8)\id-i\sin(\pi/8)P_j$, we can write
\begin{align}
\bra{x}\bra{\phi_{\mathrm{in}}} \bigg(\prod_j R(P_j)\bigg)^\dagger P_0 \bigg(\prod_j R(P_j)\bigg) \ket{x}\ket{\phi_{\mathrm{in}}} &=\sum_{P\in \mathcal{W}} \gamma_P \bra{x}\bra{\phi_{\mathrm{in}}} P\ket{x}\ket{\phi_{\mathrm{in}}} \\
&=\sum_{P\in \mathcal{W}_Z} \gamma_P \bra{x}\bra{\phi_{\mathrm{in}}} P\ket{x}\ket{\phi_{\mathrm{in}}} 
\label{eq:sumgamma}
\end{align}
for some coefficients $\gamma_P\in \mathbb{C}$. In the second equality we used the fact that $\bra{x}\bra{\phi_{\mathrm{in}}} P\ket{x}\ket{\phi_{\mathrm{in}}}=0$ unless $P$ is $Z$-type on the first $n$ qubits. From \Cref{eq:wz}, we know that each $P\in \mathcal{W}_Z$ can be written as $P' \otimes Q$, where $P' = (-1)^r\prod_{j=0}^{k} Z(\beta_j)^{u_j}$ for some bit $r\in \{0,1\}$ and string $u\in \{0,1\}^{k+1}$, and $Q\in \{I,X,Y,Z\}^{\otimes a}$. Thus we have
\begin{align}
\begin{aligned}
\bra{x}\bra{\phi_{\mathrm{in}}} P\ket{x}\ket{\phi_{\mathrm{in}}} &=(-1)^r\langle x|\prod_{j=0}^{k} Z(\beta_j)^{u_j}|x\rangle\bra{\phi_{\mathrm{in}}} Q\ket{\phi_{\mathrm{in}}}\\
&=(-1)^r\bra{\phi_{\mathrm{in}}} Q\ket{\phi_{\mathrm{in}}}\prod_{j=0}^k(\langle x|Z(\beta_j)|x\rangle)^{u_j}\\
&=(-1)^r\bra{\phi_{\mathrm{in}}} Q\ket{\phi_{\mathrm{in}}}\prod_{j=0}^{k} (1-2\XOR_{S_j}(x))^{u_j}
\label{eq:pxor}
\end{aligned}
\end{align}
where $S_j=\{j\in [n]: \beta_j=1 \}$. We have shown that each term $\bra{x}\bra{\phi_{\mathrm{in}}} P\ket{x}\ket{\phi_{\mathrm{in}}}$ appearing in \Cref{eq:sumgamma}, and therefore also the sum of all the terms,  is a polynomial function of $\{\XOR_{S_j}(x)\}_{j\in [k]}$. Using this fact in \Cref{eqn:probility-product-direct-expansion} completes the proof.
\end{proof}

\begin{proof}[Proof of \Cref{prop:Tunitary-upper-bound-PDT}]
Suppose there is a unitary Clifford+$T$ circuit $U$ which uses $k$ $T$ gates and satisfies $\Dd(\Tr_{\mathrm{anc}}[\Phi_U],U_f) \leq \epsilon$. In the standard model of unitary Clifford+$T$ circuits we would require the ancilla register to be initialized in the all-zeros state, however here we allow the ancilla register to be initialized in some (arbitrary) advice state $|\phi_{\mathrm{in}}\rangle$. Below we show that even in this potentially more powerful setting we have $k+1\geq \PDTna(f)$. This implies in particular that $\Tunitary_{\eps}(U_f)+1\geq \PDTna(f)$.

Let $x\in \{0,1\}^n$ and suppose we prepare the state $U\ket{x}\ket{\phi_{\mathrm{in}}}$ and then swap the output qubit into the first qubit and measure it in the $Z$ basis. Then the probability $p_U^{\mathrm{output}}(x)$ of the measurement outcome being $1$ satisfies
\begin{align}
|p_U^{\mathrm{output}}(x)-f(x)|
\leq \epsilon<1/2
\end{align}
since the total variation distance after measurement is upper bounded by the trace distance before measurement~\cite[Theorem 9.1]{NC10}. Moreover, by \Cref{lem:poly-form}, there exist subsets $S_0,S_1,\ldots,S_k\subseteq[n]$ and a polynomial $h$ such that
\begin{equation}
p_U^{\mathrm{output}}(x)=h\bigl(\XOR_{S_0}(x),\XOR_{S_1}(x),\ldots,\XOR_{S_k}(x)\bigr).
\end{equation}
Hence, there is a non-adaptive parity decision tree of depth $k+1$ that queries all $\{\XOR_{S_j}(x)\}_{0\leq j\leq k}$, calculates $h$, and outputs $1$ iff $h\ge 1/2$. Since $|p_U^{\mathrm{output}}(x)-f(x)|<1/2$ for every $x$, this tree computes $f(x)$ correctly on all inputs.

\end{proof}

\subsection{RPDT complexity lower bounds mixed \texorpdfstring{$T$}{T}-count}\label{sec:mixed-circuit-RPDT}
In this subsection, we extend the result from \Cref{sec:unitary-circuit-PDT} and show that for any Boolean function $f\colon\{0,1\}^n\to\{0,1\}$ and $\eps\geq 0$, its non-adaptive randomized parity decision tree complexity $\RPDTna(f)$ is upper bounded by $\Tmixed_\epsilon(U_f)+1$. 

\begin{theorem}[Part 2 of \Cref{thm:PDT-thm-intro}]\label{thm:Tmixed-lower-bound-RPDT}
For any Boolean function $f:\{0,1\}^n \to \{0,1\}$ and any $\epsilon\geq 0$, we have
\begin{equation}
\Tmixed_\epsilon(U_f) \geq \RPDTna(f) -1.
\end{equation}
\end{theorem}

To prove \Cref{thm:Tmixed-lower-bound-RPDT}, similarly to \Cref{sec:unitary-circuit-PDT}, we consider the setting of computing a Boolean function $f(x)$ using a mixed Clifford+$T$ circuit containing at most $k$ $T$ gates: the mixed circuit is a probability distribution $\{p_i\}_i$ over unitary circuits $V_i$, the input to each $V_i$ is a basis state $\ket{x}$ together with $a$ ancillas, and the output is obtained by measuring the first qubit in the $Z$ basis. As in the previous Section, we establish a slightly stronger lower bound by allowing the ancilla register to be initialized in an arbitrary advice state $|\phi_{\mathrm{in}}\rangle$. We show that there exists a randomized parity decision tree of depth $k$ that has the same output distribution as the output distribution of the mixed Clifford+$T$ circuit, and outputs 1 with probability
\begin{align}\label{eqn:mixed-circuit-probability}
p_{\mathcal{E}}^{\mathrm{output}}(x)\coloneqq\sum_i p_i(\bra{x}\otimes \bra{\phi_{\mathrm{in}}})\,V_i^\dagger \Pi\,V_i(\ket{x}\otimes \ket{\phi_{\mathrm{in}}}),
\end{align}
where $\Pi=\tfrac{\id-Z_1}{2}$. Our lower bound on $\Tmixed_\epsilon(U_f)$ then follows from the special case of this statement in which the advice state is taken to be the all-zeros computational basis state, i.e., $|\phi_{\mathrm{in}}\rangle=|0^a\rangle$.

\begin{proof}[Proof of \Cref{thm:Tmixed-lower-bound-RPDT}]
Suppose there exists a mixed Clifford+$T$ circuit which $\epsilon$-approximately implements $U_f$ using $k$ $T$ gates, i.e., a distribution $\{p_i\}_i$ over unitaries $\{V_i\}_i$, such that its associated channel $\mathcal E$ satisfies $D_\diamond(\mathcal E, U_f)\leq \epsilon$. As discussed above, in order to establish a slightly stronger result, we shall allow the ancilla register to be prepared in an arbitrary advice state $|\phi_{\mathrm{in}}\rangle$. Hence, for any $x\in\{0,1\}^n$, preparing $\ket{x}\ket{\phi_{\mathrm{in}}}$, applying $V_i$ drawn with probability $p_i$, swapping the last qubit of $V_i\ket{x}\ket{\phi_{\mathrm{in}}}$ into the first qubit and measuring it in the $Z$ basis, the probability $p^{\mathrm{output}}_\mathcal{E}(x)$ of the measurement outcome being $1$ satisfies
\begin{align}
|p^{\mathrm{output}}_{\mathcal{E}}(x)-f(x)|
\leq \epsilon
\end{align}
since the total variation distance after measurement is upper bounded by the trace distance before measurement~\cite[Theorem 9.1]{NC10}. Moreover, 
\begin{equation}
p^{\mathrm{output}}_{\mathcal{E}}(x)=\sum_i p_i\,p^{\mathrm{output}}_{V_i}(x),
\label{eq:pe}
\end{equation}
where $p^{\mathrm{output}}_{V_i}(x)$ is the output probability of the circuit $V_i$.  Now we can apply  \Cref{lem:poly-form}, which states that for each $i$ there exist subsets $S_0^{(i)},S_1^{(i)},\ldots,S_k^{(i)}\subseteq[n]$ and  polynomials $h^{(i)}$ such that, the probability $p^{\mathrm{output}}_{V_i}$ of the measurement of each circuit $V_i$ satisfies
\begin{equation}
p^{\mathrm{output}}_{V_i}(x)=h^{(i)}\left(\XOR_{S_0^{(i)}}(x),\XOR_{S_1^{(i)}}(x),\ldots,\XOR_{S_k^{(i)}}(x)\right).
\label{eq:pv}
\end{equation}
From \Cref{eq:pe} and \Cref{eq:pv} we see that  there exists a non-adaptive randomized parity decision tree of depth $k+1$ that samples $i$ according to $\{p_i\}$, queries all $\XOR_{S_j^{(i)}}(x)$, computes $p^{\mathrm{output}}_{V_i}(x)$, and outputs $1$ with probability $p^{\mathrm{output}}_{V_i}(x)$. Since $|p^{\mathrm{output}}_\mathcal{E}(x)-f(x)|\le \epsilon$ for all $x$, this randomized tree computes $f(x)$ with error at most $\epsilon$. Therefore $\RPDTna(f)\leq k+1$.  In particular, specializing to the case where $|\phi_{\mathrm{in}}\rangle=|0^a\rangle$, we can set $k=\Tmixed_\epsilon(U_f)$ and we are done.
\end{proof}

\subsection{\texorpdfstring{$T$}{T}-count upper bounds from PDT and RPDT gate complexities}
Given a non-adaptive parity decision tree $g(\XOR_{S_1}(x),\ldots,\XOR_{S_k}(x))$, we define its \emph{gate complexity} is the number of 2-input $\AND/\OR$ gates and $\mathsf{NOT}$ gates used to compute the Boolean function $g$. Analogously, given a non-adaptive randomized parity decision tree, we define its gate complexity to be the maximum gate complexity of the non-adaptive parity decision tree in the distribution.

\begin{definition}\label{def:PDT-gate-complexity}
The non-adaptive parity decision tree gate complexity of a Boolean function $f$, denoted $\gatePDT(f)$, is the minimum gate complexity among all parity decision trees that compute $f$ correctly on every input. Analogously, the non-adaptive randomized parity decision tree gate complexity,  denoted $\gateRPDT_\eps(f)$, is the minimum gate complexity of a randomized parity decision tree that outputs $f(x)$ with probability at least $1-\eps$ for all $x$.
\end{definition}

We now show that $\Tunit_{0}(U_f)$ and $\Tmix_{\eps}(U_f)$ are upper bounded by $O(\gatePDT(f))$ and $O(\gateRPDT(f))$, respectively.
\begin{theorem}[Part 3 of \Cref{thm:PDT-thm-intro}]
 For any Boolean function $f:\B^n \to \B$,
    \begin{align}
        \Tunit_0(U_f) = O(\gatePDT(f)), \qquad \Tmix_{\eps}(U_f) = O(\gateRPDT_{\eps}(f)).
    \end{align}
\end{theorem}
\begin{proof}
Note that any non-adaptive parity decision tree $g(\XOR_{S_1}(x),\ldots,\XOR_{S_k}(x))$ can be implemented by a unitary Clifford+$T$ circuit using $\Tunit_0(g)\leq O(\gatePDT(f))$ $T$ gates by \Cref{def:PDT-gate-complexity}, since each 2-input $\AND$, $\OR$, and $\mathsf{NOT}$ gate can be implemented using $O(1)$ $T$ gates exactly. Similarly, any randomized non-adaptive parity decision tree can be simulated by a mixed Clifford+$T$ circuit using $O(\gateRPDT_\epsilon(f))$ $T$ gates, whose output probability distribution is the same as the original RPDT, and thus is correct on any input with probability at least $1-\epsilon$.
\end{proof}

\subsection{The adaptive case}
In this subsection, we prove \Cref{thm:adptive-lower-PDT}, restated below:
\PDTadapt*
The proof is based on the connection between adaptive Clifford+$T$ circuits and Clifford circuits with Pauli postselections which was used in \cite{beverland2020lower} and extended in \cite{GKW24}.

Recall the notion of Pauli postselection from \Cref{eq:pps}. Note that Pauli postselection is a nonlinear operation due to the normalizing factor in \Cref{eq:pps}. Below we shall consider Clifford circuits which may include Pauli postselection and we write $C(|\psi\rangle)$ for the output state of such an operation $C$ acting on input state $|\psi\rangle$.

The following result gives a canonical form for Clifford circuit with Pauli postselections.

\begin{lemma}[Theorem A.2 of~\cite{GKW24}]\label{lem:postselect-measurement-number}
Let $C^{\post}$ be a Clifford circuit with $m$ Pauli postselections, $n$ input qubits, and $a$ ancillas.
Let $\{\ket{\phi_\lambda}\}_{\lambda\in\mathcal S}$ and $\{\ket{\psi_\lambda}\}_{\lambda\in\mathcal S}$ be two sets of $n$-qubit states indexed by $\mathcal S$.
Assume
\begin{equation}
\ket{\psi_\lambda}\ket{0^{t+a}}
  = C^{\post}\bigl(\ket{\phi_\lambda}\ket{T}^{\otimes t}\ket{0^{a}}\bigr)
  \quad\text{holds for all }\lambda\in\mathcal S,
\end{equation}
Then there exists an $(n+t)$-qubit Clifford unitary $C$ and matrices $M_1,M_2,\ldots,M_c$ such that
\begin{align}
\ket{\psi_\lambda}\ket{0^{t}}
  \propto C M_c \cdots M_2 M_1 \bigl(\ket{\phi_\lambda}\ket{T}^{\otimes t}\bigr)
  \quad\text{holds for all }\lambda\in\mathcal S,
\end{align}
where each $M_j$ is $I+P_j$ for some $(n+t)$-qubit Hermitian Pauli $P_j$, and
\begin{align}
c = t - \log\bigl(\lvert\mathsf{Stab}(\{\ket{\phi_\lambda}\}_{\lambda\in\mathcal S})\rvert\bigr)
      + \log\bigl(\lvert\mathsf{Stab}(\{\ket{\psi_\lambda}\}_{\lambda\in\mathcal S})\rvert\bigr),
\end{align}
where $\mathsf{Stab}(\{\ket{\phi_\lambda}\}_{\lambda\in\mathcal S})$ and $\mathsf{Stab}(\{\ket{\psi_\lambda}\}_{\lambda\in\mathcal S})$ are the stabilizer groups of $\{\ket{\phi_\lambda}\}_{\lambda\in\mathcal S}$ and $\ket{\psi_\lambda}\}_{\lambda\in\mathcal S}$, respectively.
\end{lemma}

In the above, the stabilizer group of a \textit{set} of states $\{|v_i\rangle\}_i$ consists of all Pauli operators $P$ such that $P|v_i\rangle=|v_i\rangle$ for all $i$.

\begin{proof}[Proof of \Cref{thm:adptive-lower-PDT}]
Let $f:\{0,1\}^n\rightarrow\{0,1\}$ be given. Suppose that $\mathcal{A}$ is an adaptive Clifford+$T$ circuit that (exactly) implements the unitary $U_f$ using $\Tadapt_{0}(U_f)$ $T$ gates in expectation (on the worst-case input state). Let $|\Phi\rangle=\frac{1}{\sqrt{2}}(|00\rangle+|11\rangle)$. If we use the $2n+1$ qubit input state $|\Phi\rangle^{\otimes n}|0\rangle$ then we get an adaptive Clifford+$T$ circuit that prepares the state 
\begin{equation}
|F\rangle\equiv (U_f\otimes I)|\Phi\rangle^{\otimes n}|0\rangle=\frac{1}{\sqrt{2^n}}\sum_{y\in \{0,1\}^n}
|y\rangle|y\rangle |f(y)\rangle,
\end{equation}
where on the RHS we have grouped the qubits so that the first two $n$-qubit registers each contain one qubit from each Bell pair $|\Phi\rangle^{\otimes n}$. Moreover, this adaptive Clifford+$T$ circuit uses an expected number of $T$ gates $t\leq \Tadapt_{0}(U_f)$. Since $|\Phi\rangle$ is a stabilizer state we can prepend a Clifford to the circuit so that it acts on the all-zeros input state $|0^{2n+1}\rangle$.

Since the adaptive Clifford circuit prepares $|F\rangle$ with zero error starting from $|0^{2n+1}\rangle$, any fixed sequence of measurement outcomes, Clifford gates, and $T$ gates that occurs with nonzero probability must give rise to a final state equal to $|F\rangle$. Let us choose a sequence of measurement outcomes and unitary gates that occurs with nonzero probability and uses the minimum number of $T$ gates. This minimum number is at most the expected number $t$ used by the adaptive Clifford+$T$ circuit. Note that in order to postselect on measuring qubit $i$ in the state $|z\rangle $ (for $z\in \{0,1\}$) we can use Pauli postselection with $P=(-1)^{z} Z$. Moreover, we can implement each $T$ gate by adjoining a magic state $|T\rangle$ and applying Clifford gates and Pauli postselection (see, e.g., \Cref{eq:tgadget}). From this we infer a circuit $C^{\mathrm{post}}$ composed of Pauli postselections and Clifford gates such that
\begin{equation}
C^{\mathrm{post}}(|0^{2n+1}\rangle|T\rangle^{\otimes t}|0^a\rangle )=|F\rangle|0^{t+a}\rangle.
\label{eq:cpost}
\end{equation}
Here $a$ is the number of ancillas used, which is some positive integer. Note that for any $x\in \{0,1\}^n$ we have
\begin{equation}
(\langle \Phi|^{\otimes n}\otimes I) |x\rangle |F\rangle =2^{-n}|x\rangle |f(x)\rangle,
\label{eq:pps2}
\end{equation}
where the Bell pairs act on qubits $i$ and $n+i$ (for each $1\leq i\leq n$). Using \Cref{eq:cpost,eq:pps2} we infer that there is a circuit $D^{\mathrm{post}}$ composed of Pauli postselections and Clifford gates such that
\begin{equation}
D^{\mathrm{post}}(|x\rangle |0^{2n+1}\rangle|T\rangle^{\otimes t}|0^a\rangle )=|x\rangle|f(x)\rangle|0^{t+a+2n}\rangle.
\label{eq:dpost}
\end{equation}
Here we used Pauli postselection onto the $+1$ eigenspace of $X\otimes X$ followed by Pauli postselection onto the $+1$ eigenspace of $Z\otimes Z$ to implement the projector onto each two-qubit state $|\Phi\rangle$ appearing in \Cref{eq:pps2}. Then we apply a Clifford which maps $n$ copies of this state to $|0^{2n}\rangle$.

Now let us partition the input and output registers of \Cref{eq:dpost} so that we can use \Cref{lem:postselect-measurement-number}. In particular we consider the set of $n+1$-qubit input states
\begin{equation}
|\phi_x\rangle=|x\rangle| 0\rangle \quad x\in \{0,1\}^n
\end{equation}
and corresponding output states
\begin{equation}
|\psi_x\rangle=|x\rangle|f(x)\rangle \quad x\in \{0,1\}^n.
\end{equation}
The stabilizer group of $\{|\phi_x\rangle\}_x$ has two elements consisting of $I^{\otimes n}\otimes Z$ and the identity. The stabilizer group of $\{|\psi_x\rangle\}_x$ depends on the function $f$ but (a) contains only $Z$-type Pauli operators and (b) other than $I^{\otimes n+1}$, does not contain any operators that act as the identity on the last qubit. From these two properties we infer that the stabilizer group of $\{|\psi_x\rangle\}_x$ contains at most two elements. Therefore
\begin{equation}
-\log\bigl(\lvert\mathsf{Stab}(\{\ket{\phi_x}\}_x)\rvert\bigr)
      + \log\bigl(\lvert\mathsf{Stab}(\{\ket{\psi_x}\}_{x})\rvert\bigr)\leq 0.
\label{eq:cleqt}
\end{equation}
Applying \Cref{lem:postselect-measurement-number} we infer that there is a an $(n+1+t)$-qubit Clifford unitary $C$ and $n+1$-qubit Paulis  $P_1,P_2,\ldots, P_c$ such that
\begin{equation}
|x\rangle|f(x)\rangle|0^t\rangle \propto C \prod_{j=1}^{c}\left(I+P_j\right)|x\rangle|0\rangle|T\rangle^{\otimes t},
\label{eq:canon}
\end{equation}
where $c\leq t$ due to \Cref{eq:cleqt}. Now consider the function
\begin{equation}
g(x)\equiv \langle x|\langle 0|\langle T|^{\otimes t}\prod_{j=1}^{c} (I+P_j) C^{\dagger} \left(I-Z_{n+1}\right) C\prod_{j=1}^{c} (I+P_j) |x\rangle |0\rangle|T\rangle^{\otimes t}.
\label{eq:gdef}
\end{equation}
From \Cref{eq:canon} we see that
\begin{align}
g(x) \propto \langle x|x\rangle \cdot  \langle f(x)|(I-Z)|f(x)\rangle \cdot \langle 0^t|0^t\rangle
\end{align}
and therefore $g(x)>0$ if and only if $f(x)=1$. To complete the proof we show that there is a non-adaptive PDT which on input $x\in \{0,1\}^n$  outputs $1$ if and only if $g(x)>0$. First write $P_0\equiv C^{\dagger} Z_{n+1} C$ and let
\begin{equation}
\mathcal{W}=\langle P_0, P_1,\ldots, P_c\rangle,
\label{eq:w2}
\end{equation}
and $|\phi_{\mathrm{in}}\rangle\equiv |0\rangle |T\rangle^{\otimes t}$.
Then 
\begin{equation}
g(x)=\sum_{P\in \mathcal{W}} \gamma_P \langle x|\langle \phi_{\mathrm{in}}| P|x\rangle|\phi_{\mathrm{in}}\rangle.
\label{eq:g2}
\end{equation}
for some coefficients $\gamma_P\in \mathbb{C}$ that can be inferred from \Cref{eq:gdef}. Comparing \Cref{eq:w2,eq:g2} with \Cref{eq:w,eq:sumgamma} we see that we have arrived at an expression for $g(x)$ which is identical to \Cref{eq:sumgamma} but with $k$ replaced by $c$. We then follow the proof of \Cref{lem:poly-form} to conclude that there exist sets $S_0, S_1,\ldots, S_c\subseteq [n]$ and a polynomial $h:\{0,1\}^{c+1}\rightarrow \mathbb{R}$ such that
\begin{align}
g(x)=h(\XOR_{S_0}(x),\XOR_{S_1}(x),\ldots, \XOR_{S_c}(x)),
\end{align}
and therefore there is a non-adaptive PDT of size $c+1\leq t+1\leq \Tadapt_{0}(U_f)+1$ that decides if $g(x)>0$ (equivalently, $f(x)=1$).
\end{proof}

\section{Examples}\label{sec:examples}

We now justify the bounds in \Cref{tab:function_costs}, starting with the upper bounds. We establish upper bounds of $O(1)$ for constant $\eps$, which can be boosted to $O(\log(1/\eps))$ for any $\eps>0$ as in \Cref{thm:fourier1norm}. Our upper bounds are either direct reductions to $\OR$ or use the upper bound of $\gateRPDT$ from \Cref{thm:PDT-thm-intro}.

\begin{itemize}
    \item[] $\OR_n(x)$: This is \Cref{thm:Toffupper}, since $U_{\OR_{n}}$ is Clifford-equivalent to $\Toff_{n+1}$. 
    \item[] $\mathsf{HW}_n^d(x)$: We divide the input into $4d^2$ sets of equal size, and use the fact that $\leq d+1$ balls thrown into $4d^2$ buckets will most likely not have 2 balls in the same bucket~\cite[Fact 1]{HSZZ06}. Thus we can simply count how many sets have any $1$s in them, which is an $\OR$, and accept if this is larger than $d$. For constant $d$, this has constant success probability. For better $d$-dependence, see the protocols of \cite[Theorem 2]{Yao03} and \cite[Theorem 1.5]{HSZZ06}. 
    \item[] $\mathsf{HW}_n^{k,2k}(x)$: Pick a random subset of the input bits with each bit chosen with probability $1/(2k)$ and compute its $\OR$. This is a constant success probability protocol for $\mathsf{HW}_n^{k,2k}$. Alternatively, the protocol in \cite[Theorem 1.5]{HSZZ06} also works for non-adaptive RPDTs. 
    \item[] $\mathsf{CW}_n^C(x)$: Using the parity check matrix definition of a linear code, checking membership in a code $C$ is a single $\OR$ of many parities.
    \item[] $\mathsf{MEQ}_{n,m}(M)$: This is equivalent to checking if the bitwise $\XOR$ of row $i$ and row $i+1$ is all zeros for all $i\in[n-1]$. This is a single $\OR$ of $m(n-1)$ two-bit $\XOR$s.
    \item[] $\mathsf{RankOne}_{n,m}(M)$: A non-adaptive RPDT upper bound of $4$ is given in \cite[Theorem 3]{GHR25}, which is easily seen to have gate complexity $O(1)$, since it is a computation on $4$ bits.
\end{itemize}

For the lower bounds, we use the $\RNAPDT$ lower bound from \Cref{thm:PDT-thm-intro}. Since $\RNAPDT$ complexity is not as well studied, we use lower bounds from communication complexity. For any Boolean function $f\colon \{0,1\}^n \to \{0,1\}$, 
the one-way communication complexity with shared randomness $R_\epsilon^\to(f^{\XOR})$ of its associated XOR function $f^{\XOR}(x,y) = f(x \oplus y)$ is at most
the non-adaptive randomized parity decision tree complexity $\RNAPDT_\epsilon(f)$ for every  $\epsilon$ (see e.g., \cite[Page 2]{KMSY18}). The following lower bounds hold even for constant $\eps=1/3$.

\begin{itemize}
    \item[] $\mathsf{GT}_n(x,y)$: $R_{1/3}^\to\big(\mathsf{GT}^{\XOR}\big)=\Omega(n)$ as shown in \cite[Theorem 19]{miltersen1995data}.
    \item[] $\mathsf{ADD}_n(x,y)$: This is even harder than $\mathsf{GT}_n(x,y)$, because if we add (using $\mathsf{ADD}_{n+1}$) the $n+1$-bit strings $0x$ and $0\bar{y}$, where $\bar{y}$ is the complement of $y$, we get the $n+1$-bit binary representation of $2^n-1+x-y$. The most significant bit of this is $1$ if and only if $x>y$.
    \item[] $\mathsf{MAJ}_n(x)$: Its associated $\XOR$ function is to decide if the Hamming distance between Alice and Bob's strings is greater than $n/2$. This famously needs $\Omega(n)$ randomized communication, even with two-way communication, even if promised that the Hamming weight is either $<n/2 - \sqrt{n}$ or $>n/2+\sqrt{n}$~\cite{CR11}. 
\end{itemize}

Our final lower bound is slightly more involved than the ones above.
\begin{theorem}
For large enough $n$, we have $\Tmix_{1/3}(\INC_n)=\Omega(n)$.
\end{theorem}
\begin{proof}
We prove the lower bound via a reduction from the two-party communication problem \textit{Augmented Index}. In this problem, Alice gets an $n$-bit string $x$, and Bob gets an index $i \in [n]$ as well as the partial string $x_{i+1} x_{i+2} .... x_n$. Bob's goal is to output $x_i$, where the communication is restricted to be one-way from Alice to Bob. 

We design a one-way randomized protocol for $\mathsf{AugIndex}_n$ whose communication cost is at most $\Tmixed_{1/3}(\INC_{2n})+1$. Alice takes $x\in\{0,1\}^n$ and forms a new string $X\in\{0,1\}^{2n}$ by replacing each $0$ with $01$ and each $1$ with $10$. Similarly, Bob forms a new string $Y\in\{0,1\}^{2n}$ by setting the last $2(n-i)$ bits in the same way as Alice and setting the first $2i$ bits to be 0. Let $Z=X\oplus Y$. Then the last $2(n-i)$ bits of $Z$ are $0$ since Alice and Bob agree on $x_{i+1},\ldots,x_n$, and the first $2i$ bits of $Z$ coincide with those of $X$. We use $k=\max\{i:Z_i=1\}$ to denote the index of this right-most 1. Then, $k$ is even iff $x_i=0$ and is odd iff $x_i=1$. Next, consider $\neg Z$, the bitwise complement of $Z$. The first $k-1$ bits of $\neg Z$ and $\INC_{2n}(\neg Z)$ coincide, while their last $2n-k+1$ bits are $0 1^{2n-k}$ and $1 0^{2n-k}$, respectively. Hence
\begin{align}
\neg Z\oplus \INC_{2n}(\neg Z)=0^{k-1}1^{2n-k+1},
\end{align}
where we use $\neg y$ denotes the bitwise complement of any bit string $y$. Consequently,
\begin{align}
x_i=\neg\XOR(\neg Z\oplus \INC_{2n}(\neg Z)).
\end{align}
Thus, there exists a mixed Clifford+$T$ circuit with $\Tmix_{1/3}(\INC_{2n})$ $T$ gates that, acting on Alice’s and Bob’s inputs, more specifically the $2n$ qubit register that encodes $\ket{x_1\ldots x_n}\ket{0^{i}x_{i+1}\ldots x_{n}}$, outputs $x_i$ with success probability at least $2/3$.  By \Cref{thm:Tmixed-lower-bound-RPDT}, this implies 
\begin{align}
\RNAPDT_{1/3}(\mathsf{AugIndex}_n)\leq \Tmix_{1/3}(\INC_{2n})-1.
\end{align}
Hence, $\Tmix_{1/3}(\INC_{2n})-1$ bits of communication suffice for Bob to compute $x_i$ with error at most $1/3$. However, \cite[Theorem 5.1]{ba2010lower} shows that $R_{1/3}^\to(\mathsf{AugIndex}_n)=\Omega(n)$. We therefore conclude that $\Tmixed_{1/3}(\INC_{n})=\Omega(n)$.
\end{proof}

\section{Acknowledgments}
We thank Craig Gidney, Uma Girish, Bill Huggins, Tanuj Khattar, Dmitri Maslov, Alex May, and Norah Tan for helpful conversations and feedback on this project.
We acknowledge the use of Gemini and ChatGPT to search the literature and suggest proof strategies.

\bibliographystyle{alphaurl}
\bibliography{adaptive-circuits}

\newcommand{\etalchar}[1]{$^{#1}$}
\begin{thebibliography}{MNSW98}

\bibitem[BBC{\etalchar{+}}95]{BBC+95}
Adriano Barenco, Charles~H. Bennett, Richard Cleve, David~P. DiVincenzo, Norman
  Margolus, Peter Shor, Tycho Sleator, John~A. Smolin, and Harald Weinfurter.
\newblock Elementary gates for quantum computation.
\newblock {\em Phys. Rev. A}, 52:3457--3467, Nov 1995.
\newblock \href {http://dx.doi.org/10.1103/PhysRevA.52.3457}
  {\path{doi:10.1103/PhysRevA.52.3457}}.

\bibitem[BCHK20]{beverland2020lower}
Michael Beverland, Earl Campbell, Mark Howard, and Vadym Kliuchnikov.
\newblock Lower bounds on the non-{Clifford} resources for quantum
  computations.
\newblock {\em Quantum Science and Technology}, 5(3):035009, 2020.
\newblock \href {http://dx.doi.org/10.1088/2058-9565/ab8963}
  {\path{doi:10.1088/2058-9565/ab8963}}.

\bibitem[BCK14]{bhrushundi2014property}
Abhishek Bhrushundi, Sourav Chakraborty, and Raghav Kulkarni.
\newblock Property testing bounds for linear and quadratic functions via parity
  decision trees.
\newblock In {\em International Computer Science Symposium in Russia}, pages
  97--110. Springer, 2014.
\newblock \href {http://dx.doi.org/10.1007/978-3-319-06686-8_8}
  {\path{doi:10.1007/978-3-319-06686-8_8}}.

\bibitem[BG16]{bravyi2016improved}
Sergey Bravyi and David Gosset.
\newblock Improved classical simulation of quantum circuits dominated by
  {C}lifford gates.
\newblock {\em Physical review letters}, 116(25):250501, 2016.
\newblock \href {http://dx.doi.org/10.1103/PhysRevLett.116.250501}
  {\path{doi:10.1103/PhysRevLett.116.250501}}.

\bibitem[BIPW10]{ba2010lower}
Khanh~Do Ba, Piotr Indyk, Eric Price, and David~P. Woodruff.
\newblock Lower bounds for sparse recovery.
\newblock In {\em Proceedings of the twenty-first annual ACM-SIAM symposium on
  Discrete Algorithms}, pages 1190--1197. SIAM, 2010.
\newblock \href {http://dx.doi.org/10.1137/1.9781611973075.95}
  {\path{doi:10.1137/1.9781611973075.95}}.

\bibitem[BK05]{bravyi2005universal}
Sergey Bravyi and Alexei Kitaev.
\newblock Universal quantum computation with ideal {Clifford} gates and noisy
  ancillas.
\newblock {\em Physical Review A—Atomic, Molecular, and Optical Physics},
  71(2):022316, 2005.
\newblock \href {http://dx.doi.org/10.1103/PhysRevA.71.022316}
  {\path{doi:10.1103/PhysRevA.71.022316}}.

\bibitem[BPP00]{BPP00}
Joan Boyar, Ren\'e Peralta, and Denis Pochuev.
\newblock On the multiplicative complexity of {B}oolean functions over the
  basis ($\wedge$,$\oplus$,1).
\newblock {\em Theoretical Computer Science}, 235(1):43--57, 2000.
\newblock \href {http://dx.doi.org/10.1016/S0304-3975(99)00182-6}
  {\path{doi:10.1016/S0304-3975(99)00182-6}}.

\bibitem[BSS16]{bravyi2016trading}
Sergey Bravyi, Graeme Smith, and John~A Smolin.
\newblock Trading classical and quantum computational resources.
\newblock {\em Physical Review X}, 6(2):021043, 2016.
\newblock \href {http://dx.doi.org/10.1103/PhysRevX.6.021043}
  {\path{doi:10.1103/PhysRevX.6.021043}}.

\bibitem[Cam17]{Cam17}
Earl Campbell.
\newblock Shorter gate sequences for quantum computing by mixing unitaries.
\newblock {\em Phys. Rev. A}, 95:042306, Apr 2017.
\newblock \href {http://dx.doi.org/10.1103/PhysRevA.95.042306}
  {\path{doi:10.1103/PhysRevA.95.042306}}.

\bibitem[CGK17]{CGK17}
Shawn~X. Cui, Daniel Gottesman, and Anirudh Krishna.
\newblock Diagonal gates in the {C}lifford hierarchy.
\newblock {\em Phys. Rev. A}, 95:012329, Jan 2017.
\newblock \href {http://dx.doi.org/10.1103/PhysRevA.95.012329}
  {\path{doi:10.1103/PhysRevA.95.012329}}.

\bibitem[CR12]{CR11}
Amit Chakrabarti and Oded Regev.
\newblock An optimal lower bound on the communication complexity of
  gap-hamming-distance.
\newblock {\em SIAM Journal on Computing}, 41(5):1299--1317, 2012.
\newblock \href {http://dx.doi.org/10.1137/120861072}
  {\path{doi:10.1137/120861072}}.

\bibitem[FvdG99]{FvdG99}
C.A. Fuchs and J.~van~de Graaf.
\newblock Cryptographic distinguishability measures for quantum-mechanical
  states.
\newblock {\em IEEE Transactions on Information Theory}, 45(4):1216--1227,
  1999.
\newblock \href {http://dx.doi.org/10.1109/18.761271}
  {\path{doi:10.1109/18.761271}}.

\bibitem[GHR25]{GHR25}
Mika Göös, Nathaniel Harms, and Artur Riazanov.
\newblock Equality is far weaker than constant-cost communication, 2025.
\newblock \href {http://arxiv.org/abs/2507.11162} {\path{arXiv:2507.11162}}.

\bibitem[GKMR14]{gosset2014algorithm}
David Gosset, Vadym Kliuchnikov, Michele Mosca, and Vincent Russo.
\newblock An algorithm for the {T}-count.
\newblock {\em Quantum Information \& Computation}, 14(15-16):1261--1276, 2014.
\newblock \href {http://dx.doi.org/10.26421/QIC14.15-16-1}
  {\path{doi:10.26421/QIC14.15-16-1}}.

\bibitem[GKW24]{GKW24}
David Gosset, Robin Kothari, and Kewen Wu.
\newblock Quantum state preparation with optimal {T}-count, 2024.
\newblock To be presented at the 2026 Annual ACM-SIAM Symposium on Discrete
  Algorithms (SODA 2026).
\newblock \href {http://arxiv.org/abs/2411.04790} {\path{arXiv:2411.04790}}.

\bibitem[Gro97]{Gro97}
Vince Grolmusz.
\newblock On the power of circuits with gates of low {L}$_1$ norms.
\newblock {\em Theoretical Computer Science}, 188(1):117--128, 1997.
\newblock \href
  {http://dx.doi.org/https://doi.org/10.1016/S0304-3975(96)00290-3}
  {\path{doi:https://doi.org/10.1016/S0304-3975(96)00290-3}}.

\bibitem[GSJ24]{gidney2024magic}
Craig Gidney, Noah Shutty, and Cody Jones.
\newblock Magic state cultivation: growing {T} states as cheap as {CNOT} gates,
  2024.
\newblock \href {http://arxiv.org/abs/2409.17595} {\path{arXiv:2409.17595}}.

\bibitem[Has17]{Has17}
Matthew~B. Hastings.
\newblock Turning gate synthesis errors into incoherent errors.
\newblock {\em Quantum Information and Computation}, 17(5–6):488–494, March
  2017.
\newblock \href {http://dx.doi.org/10.26421/QIC17.5-6-7}
  {\path{doi:10.26421/QIC17.5-6-7}}.

\bibitem[HSZZ06]{HSZZ06}
Wei Huang, Yaoyun Shi, Shengyu Zhang, and Yufan Zhu.
\newblock The communication complexity of the {H}amming distance problem.
\newblock {\em Information Processing Letters}, 99(4):149--153, 2006.
\newblock \href {http://dx.doi.org/https://doi.org/10.1016/j.ipl.2006.01.014}
  {\path{doi:https://doi.org/10.1016/j.ipl.2006.01.014}}.

\bibitem[KLM{\etalchar{+}}23]{KLMPP23}
Vadym Kliuchnikov, Kristin Lauter, Romy Minko, Adam Paetznick, and Christophe
  Petit.
\newblock Shorter quantum circuits via single-qubit gate approximation.
\newblock {\em {Quantum}}, 7:1208, December 2023.
\newblock \href {http://dx.doi.org/10.22331/q-2023-12-18-1208}
  {\path{doi:10.22331/q-2023-12-18-1208}}.

\bibitem[KM93]{KM93}
Eyal Kushilevitz and Yishay Mansour.
\newblock Learning decision trees using the {F}ourier spectrum.
\newblock {\em SIAM Journal on Computing}, 22(6):1331--1348, 1993.
\newblock \href {http://dx.doi.org/10.1137/0222080}
  {\path{doi:10.1137/0222080}}.

\bibitem[KMSY18]{KMSY18}
Sampath Kannan, Elchanan Mossel, Swagato Sanyal, and Grigory Yaroslavtsev.
\newblock {Linear Sketching over $\mathbb{F}_2$}.
\newblock In Rocco~A. Servedio, editor, {\em 33rd Computational Complexity
  Conference (CCC 2018)}, volume 102 of {\em Leibniz International Proceedings
  in Informatics (LIPIcs)}, pages 8:1--8:37, Dagstuhl, Germany, 2018. Schloss
  Dagstuhl -- Leibniz-Zentrum f{\"u}r Informatik.
\newblock \href {http://dx.doi.org/10.4230/LIPIcs.CCC.2018.8}
  {\path{doi:10.4230/LIPIcs.CCC.2018.8}}.

\bibitem[KN96]{KN96}
Eyal Kushilevitz and Noam Nisan.
\newblock {\em Communication Complexity}.
\newblock Cambridge University Press, 1996.
\newblock \href {http://dx.doi.org/10.1017/CBO9780511574948}
  {\path{doi:10.1017/CBO9780511574948}}.

\bibitem[MNSW98]{miltersen1995data}
Peter~Bro Miltersen, Noam Nisan, Shmuel Safra, and Avi Wigderson.
\newblock On data structures and asymmetric communication complexity.
\newblock {\em Journal of Computer and System Sciences}, 57(1):37--49, 1998.
\newblock \href {http://dx.doi.org/10.1006/jcss.1998.1577}
  {\path{doi:10.1006/jcss.1998.1577}}.

\bibitem[MR95]{MR95}
Rajeev Motwani and Prabhakar Raghavan.
\newblock {\em Randomized Algorithms}.
\newblock Cambridge University Press, 1995.
\newblock \href {http://dx.doi.org/10.1017/CBO9780511814075}
  {\path{doi:10.1017/CBO9780511814075}}.

\bibitem[NC10]{NC10}
Michael~A. Nielsen and Isaac~L. Chuang.
\newblock {\em Quantum Computation and Quantum Information: 10th Anniversary
  Edition}.
\newblock Cambridge University Press, 2010.
\newblock \href {http://dx.doi.org/10.1017/CBO9780511976667}
  {\path{doi:10.1017/CBO9780511976667}}.

\bibitem[RS16]{ross2014optimal}
Neil~J. Ross and Peter Selinger.
\newblock Optimal ancilla-free {Clifford+T} approximation of z-rotations.
\newblock {\em Quantum Information and Computation}, 16(11–12), 2016.
\newblock \href {http://dx.doi.org/10.26421/QIC16.11-12-1}
  {\path{doi:10.26421/QIC16.11-12-1}}.

\bibitem[RW05]{rosgen2005hardness}
Bill Rosgen and John Watrous.
\newblock On the hardness of distinguishing mixed-state quantum computations.
\newblock In {\em 20th Annual IEEE Conference on Computational Complexity
  (CCC'05)}, pages 344--354. IEEE, 2005.
\newblock \href {http://dx.doi.org/10.1109/CCC.2005.21}
  {\path{doi:10.1109/CCC.2005.21}}.

\bibitem[Wat18]{Wat18}
John Watrous.
\newblock {\em The Theory of Quantum Information}.
\newblock Cambridge University Press, 2018.
\newblock \href {http://dx.doi.org/10.1017/9781316848142}
  {\path{doi:10.1017/9781316848142}}.

\bibitem[Yao03]{Yao03}
Andrew Chi-Chih Yao.
\newblock On the power of quantum fingerprinting.
\newblock In {\em Proceedings of the Thirty-Fifth Annual ACM Symposium on
  Theory of Computing}, STOC '03, page 77–81, 2003.
\newblock \href {http://dx.doi.org/10.1145/780542.780554}
  {\path{doi:10.1145/780542.780554}}.

\end{thebibliography}

\end{document}